\documentclass{article}

\usepackage{PRIMEarxiv}

\usepackage[utf8]{inputenc} 
\usepackage[T1]{fontenc}    

\usepackage{url}            
\usepackage{booktabs}       
\usepackage{amsfonts}       
\usepackage{nicefrac}       
\usepackage{microtype}      
\usepackage{lipsum}
\usepackage{fancyhdr}       
\usepackage{graphicx}       

\usepackage{amsmath,amssymb}
\usepackage[colorlinks,bookmarks,urlcolor=blue,citecolor=blue,linkcolor=blue]{hyperref}
\usepackage{siunitx}
\newtheorem{theorem}{Theorem}[section]
\newtheorem{lemma}[theorem]{Lemma}

\newenvironment{proof}{\paragraph{Proof:}}{\hfill$\square$}

 \newcommand{\lsup}[1]{\underset{#1\to 0}{\overline{\lim}}}

\newcommand{\norm}[1]{\left\| #1 \right\|}

\title{Extreme first passage times for populations of identical rare events
}

\author{
  James MacLaurin \\
  Department of Mathematical Sciences \\
  New Jersey Institute of Technology \\
  Newark USA\\
   maclauri@njit.edu\\
   \And
 Jay Newby \\
  Department of Mathematical and Statistical Sciences \\
 University of Alberta \\
  Edmonton\\
  jnewby@ualberta.ca \\
}

\begin{document}
\maketitle

\begin{abstract}
A collection of identical and independent rare event first passage times is considered.
The problem of finding the fastest out of $N$ such events to occur is called an extreme first passage time.
The rare event times are singular and limit to infinity as a positive parameter scaling the noise magnitude is reduced to zero.
In contrast, previous work has shown that the mean of the fastest event time goes to zero in the limit of an infinite number of walkers.
The combined limit is studied.
In particular, the mean time and the most likely path taken by the fastest random walker are investigated.
Using techniques from large deviation theory, it is shown that there is a distinguished limit where the mean time for the fastest walker can take any positive value, depending on a single proportionality constant.
Furthermore, it is shown that the mean time and most likely path can be approximated using the solution to a variational problem related to the single-walker rare event.
\end{abstract}

\section{Introduction}

Stochasticity is an ingredient in a mathematical model best used when it imbues a dynamical system with behavior that a simpler deterministic model lacks. 
This is best seen when the connection between stochastic and deterministic is explicit through a law-of-large-numbers limit, in which the noise, whatever its sources, vanishes in a so called weak-noise limit.
For simplicity, we will suppose that this limit is characterized by a positive nondimensional parameter, $\epsilon > 0$, where in the limit $\epsilon \to 0^+$, the stochastic process converges to a deterministic dynamical system.
For any behavior that exists in the stochastic model but not the deterministic, this limit is singular. 
The random time over which an event that characterizes this behavior occurs---commonly called a first passage time or hitting time---diverges to infinity as $\epsilon \to 0^+$.
Events of this type are often called rare events. 
Their study is informed through exponential asymptotic analysis, namely large deviation theory \cite{freidlin1998random,dembo2009large}.

There are many situations, particularly those involving exponential proliferation, where it is not the mean first passage time of a single event that is of interest, rather it is the first out of $N$ identical rare events to occur that is relevant \cite{schuss2019redundancy}. 
For example, the dynamics of the initial outbreak in an epidemic are highly sensitive to when and where the first infection occurs \cite{WEISS1971}.
For another example, suppose that a tumor is comprised of cells with a single common ancestor. 
In this case, the first-to-arrive random walker is the first cell to accumulate enough mutations to cause uncontrolled proliferation \cite{haughey2023first}.
Another example is given by predator-prey dynamics where the prey is sufficiently skilled at evading predators that predation is a rare event \cite{kurella2015asymptotic}.
Many examples can be found in the calcium signalling literature. 
The binding of calcium to buffers in exocytosis can be characterized as a first-passage-time problem for a large number of stochastic walkers \cite{smith1998simple,matveev2018comparison}. 
Similarly calcium waves in the cell are often kicked off by the spontaneous noise-driven opening of a single channel \cite{ramlow2023integrate}.
There are many examples of bacteria entering a non proliferating dormant state to avoid the harmful effects of an antibiotic that kills the bacterium during cell division.
In some cases, switching into and out of the dormant state is governed by a rare event in a certain gene regulatory circuit \cite{newby2012isolating,newby2015bistable}. 
Once the antibiotics are no longer present, switching out of this state allows the dormant subpopulation to begin proliferating once again, driven by the first bacterium to make the switch.

The problem can be stated succinctly as follows.
Given $N$ independent and identical random walkers (assumed to be Markov processes in continuous time and space) together with a first passage time event characterized by the random variable $\tau > 0$, what is the mean of $\tau^N_{\rm first} = \min_n \{\tau_n, n=1,\ldots, N\}$?
Early work on the statistics of extreme diffusion problems such as these was done by Varadhan \cite{varadhan1967diffusion}, and it is known that $\mathbb{E}[\tau^N_{\rm first}]$ scales like $C/\log(N)$ as $N\to \infty$, where the constant $C$ depends on the statistics of the single first passage time.
This paper focuses on what happens in the dual limit $N\to \infty$ and $\epsilon \to 0^+$ (i.e. when the mean of the individual first passage times blows up).

Holcman \textit{et al} \cite{REINGRUBER2014} recently proposed extreme first passage times as a useful characterization of biological processes, particularly at the level of cellular physiology, for example, synaptic transmission \cite{matveev2018comparison} and the activation of genes by transcription factors \cite{REINGRUBER2014}. 
Lawley \textit{et al} \cite{lawley2020distribution,lawley2020probabilistic,linn2022extreme} have recently studied the extreme first-passage-time statistics of a large number of random walkers in a range of scenarios. 
Most of these studies concern situations where $N$ is the dominant asymptotic parameter. That is, it is assumed that the magnitude of the noise experienced by an individual walker is $O(1)$, and the limiting behavior of the optimal trajectory and first-hitting-time are studied as $N\to\infty$, often using the theory of extreme value statistics \cite{bray2013persistence}. A problem with this approach is that it implies that the first-hitting-particle almost instantaneously hits the target. It goes there so quickly that the effect of advection is negligible. Indeed some scholars have argued that large jumps of Brownian motions such as this are unphysical \cite{keller2004diffusion}.

In fact, Lawley and Madrid \cite{madrid2020competition} have extended these works to also consider the effect of small noise.  Relevant to our paper, they argue that the relative benefit of increased numbers of walkers is relatively weak (scaling as $1/\log(N)$), and so it remains unclear if many biological systems are in this regime. We also note the work of Weiss \cite{weiss1983order}, who determined first-hitting-time estimates when the walkers are initially uniformly distributed throughout the domain.

Let us assume that the single walker rare event starts at a deterministic attractor and terminates at some set that is well separated from the deterministic attractor but contained within its basin of attraction.
One simple example is a 1D OU process, $dX_t = -\mu X_tdt + \sqrt{2\epsilon}dW$, where the attractor is at $x=0$ and the target is at some $x=L > 0$.
This example is sometimes referred to as escape from a potential well.

One perspective used to describe a rare event is through a so called maximum likelihood trajectory (MLT) \cite{freidlin1998random,ludwig1975,maier1993,maier1997,newby2014}.
It is often true that a single state space trajectory describes a rare event.
That is, if one takes a stochastic system at the terminus of a rare event and follows the process back in time, examining the sequence of events that preceded the final step of the rare event, it is often found that a single trajectory describes this sequence of events.
To make this more concrete, we can imagine a distribution of possible trajectories and define the MLT as the mode, or peak, of this distribution.
It is possible to formulate asymptotic approximations of MLTs as solutions to a certain variational problem \cite{freidlin1998random,dembo2009large}.

The single walker rare event can be described in two parts: a lag phase and a transition phase \cite{maier1993}.
The rare event begins with an exponentially large (in $\epsilon^{-1}$) lag phase where for the great majority of time the process fluctuates within a $\sqrt{\epsilon}$ domain around the attractor. 
During this time, there may be a few relatively large excursions away from the attractor, but they are not large enough to attain the target $L$ and are quickly pulled back to the attractor. 
After the lag phase, there is a transitional phase characterized by a MLT that starts on the boundary of the $\sqrt{\epsilon}$ neighborhood and terminates at the target. 
The total time of the transitional phase scales like $\log(1/\sqrt{\epsilon})$ as $\epsilon \to 0$. 

If we understand the MLT for the single-walker rare event, does this trajectory also describe the minimum of $N$ identical rare events?
Perhaps the extreme first passage time is an unlikely deviation from the typical rare event, an event that is doubly rare, in some sense.
It is known that (in the limit $N\to \infty$ with $\epsilon$ fixed) the MLT for an extreme event is a geodesic, relative to the geometry induced by the diffusion coefficients \cite{varadhan1967diffusion}. 
For example, simple additive white noise (constant diffusivity) results in a MLT that is a direct linear path from the starting point to the target.
This effectively means that there are so many walkers that the first to arrive is perturbed by Brownian fluctuations that send it straight to the target in a short period of time, and therefore the effect of a "drift velocity" term on the MLT is negligible once there are enough walkers, even if this velocity pushes the process away from the target.

We are left with the general possibility that for $N$ finite, the first-to-arrive walker follows the single walker MLT in the limit $\epsilon \to 0$.
On the other hand, for $\epsilon$ fixed and $N\to \infty$, the first-to-arrive walker follows a straight line path which will, in general, be very different than the single walker MLT.~~
Is it possible to describe the extreme first passage time MLT in the dual limit $\epsilon \to 0$ and $N\to\infty$?
Does the trajectory smoothly transform from the single walker MLT to the straight line path for sufficiently large $N$?

Another perspective on the problem focuses on the asymptotic value of $\mathbb{E}[\tau^N_{\rm first}]$.
Suppose we have a rare event such that $\tau_n \simeq \text{Exponential}(\lambda_{\epsilon}),$ asymptotically as $\epsilon \to 0^+$. 
Previous work on this problem identified two regimes \cite{madrid2020competition}.
For fixed $N > 0$, $\mathbb{E}[\tau^N_{\rm first}] \propto \frac{1}{\lambda_{\epsilon} N}$ as $\epsilon \to 0^+$. 
For fixed $\epsilon > 0$, we know that $\mathbb{E}[\tau^N_{\rm first}]$ scales like $1/\log(N)$ as $N\to \infty$. 
We are interested in the crossover between the $1/N$ regime and the $1/\log(N)$ regime.

In this paper, we will show that if we combine these two perspectives a single picture emerges wherein there are three regimes for the MLTs and $\mathbb{E}[\tau^N_{\rm first}]$ in the double limit $\epsilon\to 0$ and $N\to \infty$.
For smaller values of $N$ we have the first regime where the first-to-arrive walker follows the single walker MLT.~~  
In this case, $\mathbb{E}[\tau^N_{\rm first}]$ is decreased like $1/N$ by eliminating time spent in the exponentially long lag phase where the walker fluctuates near the attractor.
As $N$ is increased sufficiently to eliminate the lag phase, we enter the second regime.
Reducing $\mathbb{E}[\tau^N_{\rm first}]$ now requires shortening the duration of the transition phase of the rare event.
The extreme MLT smoothly deforms away from the single walker MLT in order to speed up the transition time.
Eventually, for $N$ sufficiently large, it is no longer possible to improve $\mathbb{E}[\tau^N_{\rm first}]$ by deforming the extreme MLT because it is a straight line from the attractor to the target.
In this regime, $\mathbb{E}[\tau^N_{\rm first}]$ is reduced like $1/\log(N)$ by a transition path that moves faster and faster along the direct linear extreme MLT.

To our knowledge, the second regime described above was previously unknown. 
Moreover, we find that in the second regime, we can pick any $\tau_{\infty} > 0$ and it is always possible to take the double limit $\epsilon\to 0$ and $N\to \infty$ in such a way that $\mathbb{E}[\tau^N_{\rm first}] \to \tau_{\infty}$.

The remainder of the paper is organized as follows. 
First, in Section \ref{sec:formulation}, we present a precise formulation of the extreme first passage time problem.
Second, we include some necessary background from large deviation theory in Section \ref{sec:background}.
We then present our main result in Section \ref{sec:main}.
Finally, we illustrate our main result in Section \ref{sec:example} with a simple example using an Ornstein-Uhlenbeck process where many exact results are obtainable.
We conclude by discussing phenomena that are known to exist in general nonlinear systems along with numerical strategies for computing MLTs.

\section{Mathematical formulation of the extreme first hitting time problem}
\label{sec:formulation}

Consider the diffusion of $N$ independent random walkers in $\mathbb{R}^d$ (continuous in time and space), each satisfying the stochastic differential equation
\begin{equation}
\label{eq:sde_gen}
dX^j_t = f( X^j_t)dt + \sqrt{2\epsilon} dW^j_t,
\end{equation}
with $1\leq j \leq N$. Here $W^j_t := (W^{1,j}_t, \ldots, W^{d,j}_t)$ is $d$-dimensional Brownian Motion. It is assumed that $0$ is a globally-attracting fixed point of the dynamics (when $\epsilon = 0$). The initial condition is assumed to be zero for every walker. 

The main goal of this paper is to determine the probabilistic distribution of the first-hitting-time of some small target region centered at $L \in \mathbb{R}^d$. In particular, we wish to understand how the expected first hitting time depends on $\epsilon$ (which is asymptotically small) and $N$ (which is asymptotically large). Concurrently, we wish to understand the most likely trajectory following by the first walker to hit the target.

Define a small ball about $L$ to be
\begin{equation}
B_{\delta}(L) = \big\lbrace z \in \mathbb{R}^d : \| z - L \| \leq \delta \big\rbrace ,
\end{equation}
for some $\delta = \delta(\epsilon ,N)$ that converges to zero very slowly as $\epsilon \to 0$ and $N\to\infty$. As long as $\delta$ converges to 0 relatively slowly, the leading order asymptotic estimate of the first hitting-time will only depend on $\epsilon$ and $N$, and it won't depend on the specific rate that $\delta \to 0^+$. 
time.

Let $\tau^N_{\rm first}$ be the first hitting time of $L$ over all $N$ walkers, and let $\tau_j$ be the first hitting time of a single walker. Then, we have
\begin{align}
\label{eq:fpt_def1}
\tau_j &= \inf\big\lbrace t \geq 0: X^j_t \in  B_{\delta}(L) \big\rbrace, \\
\label{eq:fpt_def2}
\begin{split}
\tau^N_{\rm first} 
&= \inf \left\{ \tau_1, \ldots, t_N \right\}.
\end{split}
\end{align}
Throughout this paper we assume that $f: \mathbb{R}^d \to \mathbb{R}^d$ is Lipschitz, i.e. there exists a constant $C_f > 0$ such that for all $x,y \in \mathbb{R}^d$
\begin{align}\label{eq: f lipschitz}
\norm{f(x) - f(y)} \leq C_f \norm{x-y}
\end{align}

\section{Background on Large Deviations Theory}
\label{sec:background}
Classically, the asymptotic behavior of stochastic processes under weak noise conditions (including first-hitting-times) has been studied using Freidlin-Wentzell Theory (also referred to as the Theory of Large Deviations \cite{dembo2009large}). We now briefly recap the relevant aspects of this theory (mostly following the treatment in \cite{dembo2009large}). The approach builds upon an asymptotic description of the probability law for trajectories that have a fixed (non random) start and end time at $t=0$ and $t=T$, respectively.
These trajectories can also be conditioned on the initial and final position (i.e., $X_0$ and $X_T$).

For some fixed and non random time $T > 0$, define the rate function $I_T$ that gives the asymptotic probability of a given trajectory $X_t$, within the time window $0\leq t \leq T$. Specifically, $I_T$ is a map from $\mathcal{C}([0,T],\mathbb{R}^d)$ to $\mathbb{R}^+$. It is stipulated that $I_T(z) < \infty$ only if $z\in H_1$ (the space of all functions in $\mathcal{C}\big( [0,T],\mathbb{R}^d \big)$ which possess a bounded derivative for Lebesgue-almost-every time). In this case, the rate function is defined to be
\begin{equation}
\label{eq:ratefn}
I_T(y) =   \int_0^T L(y(t), \dot{y}(t))dt,
\end{equation}
where
\begin{equation}
\label{eq:lagrangian}
L(y, \dot{y}) =  \frac{1}{4}\left( \dot{y} - f(y) \right)^T\left( \dot{y} - f(y) \right) .
\end{equation}
$L(y,\dot{y})$ is called the Lagrangian. Formally, the Large Deviations Principle expresses the asymptotics of certain events in the following manner.  Suppose that $\mathcal{A}, \mathcal{O} \subset \mathcal{C}\big( [0,T] , \mathbb{R}^d \big)$ are (respectively) open and closed subsets of the space of all continuous functions, whose value at time $0$ is $0$. The topology on $\mathcal{C}\big( [0,T] , \mathbb{R}^d \big)$ is the topology of uniform convergence---this means that every open set must be able to be written as a union (potentially infinite) of small balls of the form
\begin{align}
\bigg\lbrace z \in \mathcal{C}\big( [0,T] , \mathbb{R}^d \big) \; : \; \sup_{t\in [0,T]} \sup_{1\leq j \leq d} \big| z^j(t) - y^j(t) \big| < \xi \bigg\rbrace,
\end{align}
for some $\xi > 0$ and $y \in \mathcal{C}\big( [0,T] , \mathbb{R}^d \big)$. We let $X^j_{[0,T]} \in \mathcal{C}\big( [0,T],\mathbb{R}^d \big)$ denote the entire path of the $j^{th}$ walker. The Large Deviations results originally due to Freidlin and Wentzell state that \cite{dembo2009large,freidlin1998random} for any $1\leq j \leq N$ we have
\begin{align}
\lim_{\xi \to 0^+} \sup_{\epsilon \leq \xi} \epsilon \log \mathbb{P}\big( X^j_{[0,T]} \in \mathcal{A}\big) \leq - \inf_{z \in \mathcal{A}} I_T(z), \label{FW upper bound}  \\
\lim_{\xi \to 0^+} \inf_{\epsilon \leq \xi} \epsilon \log \mathbb{P}\big( X^j_{[0,T]} \in \mathcal{O}\big) \geq - \inf_{z \in \mathcal{O}} I_T(z). \label{FW lower bound} 
\end{align}

\subsection{Large Deviation Theory of maximum likelihood trajectories}
\label{sec:LDT MLT}
Next, we outline how the Large Deviations results can be used to estimate MLTs (for a single walker, in the small noise limit). The theory in this sub-section is already well-established (see \cite{grafke2019numerical} for a thorough review). Define the Hamiltonian $\mathcal{H}: \mathbb{R}^d \times \mathbb{R}^d \to \mathbb{R}$ as
\begin{align}
\label{eq:ham_ldt}
\mathcal{H}(x, p) = \sup_{z \in \mathbb{R}^d} \left\lbrace z^T p  - L(x, z)) \right\rbrace = p^T f(x) + p^Tp,
\end{align}
where $L(x, z)$ is defined by \eqref{eq:lagrangian}.
Because the Lagrangian is convex in its second argument, Legendre duality implies that 
\begin{align}
L(y, z) = \sup_{p \in \mathbb{R}^d} \left\lbrace  z^T p - \mathcal{H}(x, p)  \right\rbrace ,
\end{align}
which is consistent with \eqref{eq:lagrangian}.
Given the process $X^j_t$ defined by \eqref{eq:sde_gen}, suppose that we wish to estimate (i) the probability of $X^j_t = B_{\delta}(z)$ for some $z\in \mathbb{R}^d$, and (ii) the most likely trajectory followed by the walker in attaining this value. In other words, we wish to minimize the cost function,
\begin{align}
\label{eq:cost_fn}
V(z,t) = \inf_{y \in H_1 : y(t) = z} I_t(y),
\end{align}
up to time $t$ for $z\in \mathbb{R}^d$, and find the trajectory $\varphi \in H_1$ that attains this minimum. This means that $V(z,t)$ is such that, for $\delta \ll 1$,
\begin{align}
\mathbb{P}\big(  \| X^j_t - z \| \leq \delta \big) \simeq \exp\big( - \epsilon^{-1} V(z,t) \big) .   
\end{align}
It is well known (see Chapter 5 of Dembo and Zeitouni \cite{dembo2009large}) that the cost function satisfies the following Hamilton-Jacobi equation,
\begin{equation}
\frac{\partial}{\partial t}V(x,t)= \mathcal{H}(x, \nabla_x V ).
\label{eq: Hamiltonian PDE} 
\end{equation}
Finally, taking $t\to\infty$ we obtain the so-called quasi-potential $V: \mathbb{R}^d \to \mathbb{R}^+$,
\begin{equation}
\label{eq:quasi potential}
V(z) = \inf_{t\geq 0} \inf_{\stackrel{y\in H_1 }{y(t) = z}} I_t(z)  = \lim_{t\to\infty} V(z, t).
\end{equation}
The quasipotential is a stationary solution to the Hamilton-Jacobi equation \eqref{eq: Hamiltonian PDE}.

To emphasize trajectories that are conditioned on a fixed final time $T$ and position $L$, we introduce the notation $\varphi_{t|T}$, where it is understood that it is a function of time $t$, where $0 \leq t \leq T$.
By Hamilton's Principle, the path $\varphi_{t|T}$ that minimizes the rate functional \eqref{eq:ratefn} satisfies
\begin{equation}
    \label{eq:ham_dyn_gen}
    \begin{split}
       \dot{\varphi}_{t|T} &= \nabla_p \mathcal{H}(\varphi_{t|T}, p_{t|T}), \\
       \dot{p}_{t|T} &= -\nabla_x \mathcal{H}(\varphi_{t|T}, p_{t|T}),
    \end{split}
\end{equation}
along with the boundary conditions,
\begin{equation}
    \varphi_{0|T} = 0  , \; \; \varphi_{T|T} = L.
\end{equation}
Note that the initial value $p_{0|T}$ is not a free parameter and must be used to satisfy $\varphi_{T|T} = L$. 
See \cite{grafke2019numerical} for a review of numerical methods one may employ to compute solutions of \eqref{eq:ham_dyn_gen}. 
An important property of the system \eqref{eq:ham_dyn_gen} is that it has a constant of motion. 
In particular, we have that $\mathcal{H}(\varphi_{t|T}, p_{t|T}) \equiv \text{const}$ over any solution of \eqref{eq:ham_dyn_gen}.

In addition to satisfying the Hamilton-Jacobi equation \eqref{eq: Hamiltonian PDE}, the cost function \eqref{eq:cost_fn} can be integrated along the optimal paths satisfying \eqref{eq:ham_dyn_gen}. 
We define the cost function along an MLT with
\begin{equation}
\label{eq:S_V}
    S_{t|T} = V(\varphi_{t|T}, t).
\end{equation}
Differentiating the cost function with respect to time along an optimal path yields
\begin{equation}
\begin{split}
    \dot{S} = \frac{d}{dt}V(\varphi_{t|T}, t) &= (\nabla_x V(\varphi_{t|T}, t))^T \frac{\partial\varphi_{t|T}}{\partial t} + \frac{\partial V}{\partial t}(\varphi_{t|T}, t) \\
    &= p_{t|T}^T\nabla_p \mathcal{H}(\varphi_{t|T}, p_{t|T}) - \mathcal{H}(\varphi_{t|T}, p_{t|T}) \\
    &=  p_{t|T}^Tp_{t|T}.
\end{split}
\end{equation}
Hence, we can supplement the system \eqref{eq:ham_dyn_gen} with the auxiliary equation,
\begin{equation}
\label{eq:S_of_t}
    \dot{S}_{t|T} = p_{t|T}^Tp_{t|T}, \quad S_{0|T} = 0.
\end{equation}
Notice that it immediately follows from the above equation and \eqref{eq:S_V} that $\frac{d}{dt}V \geq 0$ along optimal paths.
The quasi-potential \eqref{eq:quasi potential} is given by the "zero energy" solutions to \eqref{eq:ham_dyn_gen}. 
In particular, the  $\mathcal{H}(\varphi_{t|T}, p_{t|T}) = 0$ solutions to \eqref{eq:ham_dyn_gen} are consistent with the steady state solution to the Hamilton-Jacobi equation \eqref{eq: Hamiltonian PDE}.

There are several more important properties of the Hamiltonian dynamical system \eqref{eq:ham_dyn_gen} that are relevant in our current setting. 
For $x$ in a neighborhood of a stable attractor of the deterministic dynamics ($\dot{\varphi} = f(\varphi))$, the quasi-potential is determined by optimal paths that converge backward in time to the stable attractor as $t - t_0 \to \infty$. 
To see this we first notice that $\nabla_p\mathcal{H}(x, 0) = f(x)$, which means that the deterministic flows exist on a manifold of the Hamiltonian dynamical system \eqref{eq:ham_dyn_gen}. 
Since, $\nabla_x \mathcal{H}(x, 0) \equiv 0$ for all $x$, it follows that deterministic fixed points, $x_*$, satisfying $f(x_*) = 0$, are also fixed points of the Hamiltonian dynamical system. 
(Note that it is a consequence of Liouville's Theorem that there are no stable or unstable attractors in the Hamiltonian dynamical system, only saddles and centers.) 
Finally, since $\mathcal{H}(x_*, 0) = 0$, the Hamiltonian function must be zero along all trajectories converging to or diverging from $(x, p) = (x_*, 0)$. 

Together, these properties lead to the following understanding of the zero-energy optimal paths. 
If the starting point is the attractor $x_*$ and the end point is $L$ (within the basin of attraction) then $\dot{S}_{t|T} \geq 0$ and the optimal paths describe MLTs.~~
If, on the other hand, we switch the start and end points so the trajectory starts at $L$ and ends at $x_*$ then $\dot{S}_{t|T} \equiv 0$, $p_{t|T}\equiv 0$, and $\varphi_{t|T}$ is a solution to the deterministic dynamics.

\section{Main result}
\label{sec:main}

In this section, we consider the first passage time problem \eqref{eq:fpt_def2}.
Unless otherwise stated, it is assumed that the distinguished limit $\epsilon \to 0^+$ and $N\to \infty$ is taken with 
\begin{equation}
\label{eq:distinguished limit}
    \Lambda = \epsilon \log(N)
\end{equation} 
fixed constant.
We will show that there are three different regimes for $\tau^N_{\rm first}$, depending on the relative magnitude of $\Lambda$. 
Note that in Regime 3, we consider the case where $N\to\infty$ fast enough compared to $\epsilon\to 0^+$ so that $\epsilon \log(N) \to \infty$, which we can think of as letting $\Lambda \to \infty$.
\begin{itemize}
\item {\bf Regime 1: Labor shortage.}
$\Lambda \leq V(L)$ so that $N \leq N_* = \exp( V(L) \epsilon^{-1})$.  
In this regime, increasing the number of walkers (i.e., increasing "labor") has the strongest effect on reducing the extreme first hitting time. 
By labor "shortage" we mean that $\mathbb{E}[\tau^N_{\rm first}]$ diverges to $\infty$.
The trajectory followed by the first-arriving walker is with high probability approximately the same as in the single-walker case (i.e. when $N=1$). The expected time $\mathbb{E}[\tau^N_{\rm first}]$ that it takes the first walker to arrive is the expected time for the single-walker case, divided by $N$. Furthermore, $\mathbb{E}[\tau^N_{\rm first}]$ diverges to $\infty$ as $\epsilon \to 0^+$, with rate
\begin{equation}
\lim_{\delta \to 0^+}\lim_{\epsilon\to 0^+} \epsilon \log \mathbb{E}[\tau^N_{\rm first}] = V(L) - \Lambda.
\end{equation}

\item {\bf Regime 2. Balanced labor.}
$V(L) < \Lambda < \infty$ so that $N > N_* = \exp( V(L) \epsilon^{-1})$. 
In this regime, there is a sufficient number of walkers to accomplish the extreme rare event in finite time as $\epsilon \to 0^+$.
The trajectory followed by the first-arriving walker is with high probability different from the trajectory in the single-walker case. Furthermore the expected time that it takes the first walker to arrive does not diverge as $N\to\infty$ and $\epsilon \to 0^+$. Indeed, the expected time $\mathbb{E}[\tau^N_{\rm first}]$ that it takes the first walker to arrive is $\mathbb{E}[\tau^N_{\rm first}] \simeq \Psi_L(\Lambda)$, where 
\begin{align} 
\label{eq: Psi x c}
\Psi_L(\Lambda) = \inf\big\lbrace t\geq 0 \; : \; V(L,t) = \Lambda \big\rbrace,
\end{align}
and $V$ is the cost function defined by \eqref{eq:cost_fn} and \eqref{eq: Hamiltonian PDE}.
In other words, $\Psi_L(\Lambda)$ can be thought of as the shortest time that it takes for the cost-function to reach $\Lambda$.

\item {\bf Regime 3. Labor surplus.}
We assume that the number of walkers $N := N(\epsilon)$ is such that $\epsilon \log (N(\epsilon)) \to \infty$.
In this regime, adding more walkers has the smallest effect on the extreme first hitting time. There are so many walkers that the first-to-arrive walker almost instantaneously travels straight to the target. 
The expected first hitting-time goes to $0$: more precisely with high probability it scales as
\begin{align}
\lim_{\delta \to 0^+}\lim_{\epsilon \to 0^+}\epsilon\log( N )\tau^N_{\rm first} \to \frac{\norm{L}^2}{4}.
\end{align}
The first walker moves to the target so quickly that the advection term $f$ has a negligible effect on the most likely trajectory. 
\end{itemize}
We emphasize again that we are taking $\delta \to 0^+$ relatively slowly compared to $\epsilon \to 0^+$. 
In fact, the asymptotic results in this section are unchanged for any $\delta := \delta(\epsilon)$ such that the Lebesgue Measure $\mu_{Leb}$ of the ball $B_{\delta(\epsilon)}$ about $L$ decays slowly enough; that is,
\begin{equation}
    \lim_{\epsilon \to 0^+}\epsilon \log (\mu_{Leb}\big( B_{\delta(\epsilon)}(L))\big)=0.
\end{equation}
However because the proofs are already quite technical, we avoid writing the explicit dependence of $\delta$ on $\epsilon$ and simply write that $\epsilon\to 0^+$ first, and then $\delta \to 0^+$ second.



\subsection{Proof in the case of a diverging first-hitting-time (Regime 1)}

Lets first consider the case that $\Lambda < V(L)$. We fix the number of walkers to be $N = \lfloor \exp( \epsilon^{-1} \Lambda ) \rfloor$ and we (i) asymptote $\epsilon \to 0^+$ before (ii) asymptoting $\delta \to 0^+$.

Our method of proof is to show that the first-hitting time satisfies an exponential concentration inequality. More precisely, we will prove that for any $\eta > 0$, there exists $\bar{\delta}_{\eta},C_{\eta} > 0$ such that for all $\delta < \bar{\delta}_{\eta}$ and $\bar{\epsilon}_{\delta,\eta}$ such that for all $\epsilon < \bar{\epsilon}_{\delta,\eta}$
\begin{align}
\epsilon \log \mathbb{P}\big( \big|  \epsilon \log \tau^N_{\rm first} - \epsilon \log   \mathbb{E}[\tau^N_{\rm first}] \big|   \geq \eta \big)  &< - C_{\eta} , \label{eq: epsilon log bound concentration}\\
\lim_{\delta \to 0^+}\lsup{\epsilon}\epsilon \log \mathbb{P}\big( \big| \epsilon \log (\tau^N_{\rm first}) - \epsilon \log\mathbb{E}[\tau^N_{\rm first}] \big| \leq \eta \big) &= 0.\label{eq: corollary 4.1}
\end{align}
The above results imply convergence in probability.

In order that \eqref{eq: epsilon log bound concentration} holds, it will suffice to prove the following Lemma. One easily checks that \eqref{eq: epsilon log bound 1} - \eqref{eq: infinite epsilon log bound 2} imply the concentration inequality in \eqref{eq: epsilon log bound concentration}. Concerning \eqref{eq: corollary 4.1}, if it were not true, then we would find that the total probability mass decays exponentially (which is a contradiction).
\begin{lemma}\label{Regime 1 lemma}
For any $\zeta > 0$,
\begin{align}
\label{eq: epsilon log bound 1}
\lim_{\delta \to 0^+}\lsup{\epsilon}\epsilon \log \mathbb{P}\big( \epsilon \log (\tau^N_{\rm first}) + \Lambda - V(L)  < -\zeta \big) &< 0 , \\
\lim_{\delta \to 0^+}\lsup{\epsilon}\epsilon \log \mathbb{P}\big( \epsilon \log (\tau^N_{\rm first}) + \Lambda - V(L)  > \zeta \big) &= -\infty. \label{eq: infinite epsilon log bound 2}
\end{align}
Furthermore,
\begin{align} \label{eq: other main result}
\lim_{\delta \to 0^+}\lim_{\epsilon \to 0^+}\epsilon \log \mathbb{E}\big[ \tau^N_{\rm first} \big] = V(L) - \Lambda .
\end{align}
\end{lemma}
\begin{proof}


In Section \ref{Proofs First Hitting Single Walker} we prove the following asymptotics for the first hitting time of a single walker: for any $ 0 \leq b \leq V(L)$ and any $a > 0$ (and typically $a \ll 1$),
\begin{align}
\lim_{\delta \to 0^+}\lim_{\zeta \to 0^+}\lsup{\epsilon}\left \vert \epsilon \log \mathbb{P}\big( \epsilon \log (\tau_1) \in [b-\zeta, b+\zeta ]\big) - b + V(L)  \right \vert &=0 , \label{eq: single walker hitting time asymptotic initial}\\
\lim_{\delta \to 0}\lsup{\epsilon}\epsilon \log \mathbb{P}\big( \epsilon \log (\tau_1) < b-a \big) &< b- V(L) , \label{eq: single walker hitting time asymptotic 2 initial} \\
\lim_{\delta \to 0}\lsup{\epsilon}\epsilon \log \mathbb{P}\big( \epsilon \log( \tau_1) > b+a \big) &< b- V(L).  \label{eq: single walker hitting time asymptotic 3 initial} 
\end{align}

Next we consider multiple walkers, exploiting the fact that they are identically distributed. To prove \eqref{eq: infinite epsilon log bound 2}, using a union of events bound,
\begin{align}
\mathbb{P}\big( \epsilon \log (\tau^N_{\rm first}) <  V(L) - \Lambda - \zeta \big) &\leq \sum_{j=1}^N \mathbb{P}\big( \epsilon \log (\tau_j) <  V(L) - \Lambda - \zeta \big) \\
&= N \mathbb{P}\big( \epsilon \log (\tau_1) <  V(L) - \Lambda - \zeta \big),\label{eq:temproary tilde tau j}
\end{align}
since the walkers are identically distributed. It thus follows from \eqref{eq: single walker hitting time asymptotic 2 initial} and \eqref{eq:temproary tilde tau j} that  \eqref{eq: epsilon log bound 1} holds, i.e.
\begin{equation}
\lim_{\delta \to 0^+} \lsup{\epsilon} \epsilon \log   \mathbb{P}\big( \epsilon \log (\tau^N_{\rm first}) <  V(L)-\Lambda - \zeta \big) < 0.
\end{equation}
 Next, we prove \eqref{eq: epsilon log bound 2} using the fact that the walkers are independent and identically distributed. Using the bound in \eqref{eq: single walker hitting time asymptotic 3 initial}, it follows that for any $\zeta > 0$ we have
 \begin{equation}
\begin{split}
\mathbb{P}\big( \epsilon \log (\tau^N_{\rm first}) >  V(L) - \Lambda + \zeta \big) &= \prod_{j=1}^N \big\lbrace 1 - \mathbb{P}\big( \epsilon \log (\tau_j )\leq  V(L) - \Lambda + \zeta \big) \big\rbrace  \\
&\leq \big\lbrace 1- \exp\big( - \epsilon^{-1}(- \Lambda - \xi) \big) \big\rbrace^N ,
\end{split}
\end{equation}
for some $\xi > 0$, as long as $\delta$ is sufficiently small. Taking logarithms of both sides, before multiplying both sides by $\epsilon$, we thus find that (making use of the inequality $\log(1-u) \leq -u$ for any $u > 0$ in the second line)
\begin{equation}
\begin{split}
 \epsilon \log \mathbb{P}\big( \epsilon \log (\tau^N_{\rm first}) >  V(L) - \Lambda + \zeta \big) 
&\leq   \epsilon N \log \big\lbrace 1- \exp\big( - \epsilon^{-1}( \Lambda - \xi) \big) \big\rbrace \\
&\leq - \epsilon N \exp\big( - \epsilon^{-1}( \Lambda - \xi) \big) \\
&= - \epsilon  \exp\big( \epsilon^{-1} \xi) \big), \label{eq: temporary epsilon bound upper infinite}
\end{split}
\end{equation}
and in the last line we have made use of the scaling of $N$ in \eqref{eq:distinguished limit}. Observe that as $\epsilon \to 0^+$, we obtain \eqref{eq: infinite epsilon log bound 2}.

Next we prove \eqref{eq: other main result}. Writing $f_{\epsilon} =  \epsilon \log (\tau^N_{\rm first}) - V(L) + \Lambda$, for any $\eta \ll 1$ we decompose the expectation as
\begin{align}
\mathbb{E}\big[ \tau^N_{\rm first} \big] = \mathbb{E}\big[ \tau^N_{\rm first} \chi\big\lbrace f_{\epsilon} \leq -\eta \big\rbrace  \big] + \mathbb{E}\big[ \tau^N_{\rm first} \chi\big\lbrace -\eta < f_{\epsilon} < \eta \big\rbrace  \big]
+\mathbb{E}\big[ \tau^N_{\rm first} \chi\big\lbrace f_{\epsilon} \geq \eta \big\rbrace  \big].\label{eq: decompose tau N first three}
\end{align}
We bound these terms separately. It is immediate from the definition that
\begin{align}\label{eq: eta 2}
\lsup{\epsilon} \epsilon \log \mathbb{E}\big[ \tau^N_{\rm first} \chi\big\lbrace f_{\epsilon} \leq -\eta \big\rbrace  \big] \leq - \eta + V(L) - \Lambda .
\end{align}
Next, we claim that for any $\eta > 0$,
\begin{equation}\label{eq: eta 1}
\lim_{\delta \to 0^+}\lsup{\epsilon} \epsilon \log \mathbb{E}\big[ \tau^N_{\rm first} \chi\big\lbrace f_{\epsilon} \geq \eta \big\rbrace  \big] = -\infty.
\end{equation}
To see this,
\begin{equation}
\begin{split}
\mathbb{E}\big[ &\tau^N_{\rm first} \chi\big\lbrace f_{\epsilon} \geq \eta \big\rbrace  \big] \leq
\sum_{m=0}^{\infty}  \mathbb{E}\big[ \tau^N_{\rm first} \chi\big\lbrace f_{\epsilon} \in \big[\eta / 2 + m \sqrt{\epsilon}, \eta /2 + (m+1)\sqrt{\epsilon} \big) \big\rbrace \big] \\
\leq & \sum_{m=0}^{\infty} \mathbb{P}\big( f_{\epsilon} \in \big[\eta / 2 + m \sqrt{\epsilon}, \eta /2 + (m+1)\sqrt{\epsilon} \big) \big)\exp\big( \epsilon^{-1}(V(L) - \Lambda) +  \epsilon^{-1/2}(m+1) + \epsilon^{-1} \eta / 2 \big)   \\
\leq &   \sum_{m=0}^{\infty} \exp\bigg( \epsilon^{-1}(V(L) - \Lambda) +  \epsilon^{-1/2}(m+1)  -\epsilon\exp\big(\epsilon^{-1}\eta /2 + \epsilon^{-1/2}m \big) + \epsilon^{-1} \eta / 2 \bigg) ,
\end{split}
\end{equation}
making use of \eqref{eq: temporary epsilon bound upper infinite}, as long as $\delta$ is sufficiently small. Now 
\begin{align}
\lim_{\epsilon \to 0^+} \epsilon \log  \sum_{m=0}^{\infty} \exp\bigg( \epsilon^{-1}(V(L) - \Lambda) +  \epsilon^{-1/2}(m+1)  -\epsilon\exp\big(\epsilon^{-1}\eta /2 + \epsilon^{-1/2}m \big)+ \epsilon^{-1} \eta / 2  \bigg) \nonumber \\
= V(L) - \Lambda + \eta/2 + \lim_{\epsilon \to 0^+} \epsilon \log  \sum_{m=0}^{\infty} \exp\bigg( \epsilon^{-1/2}(m+1)  -\epsilon\exp\big(\epsilon^{-1}\eta /2 + \epsilon^{-1/2}m \big) \bigg)     \nonumber \\
:=  V(L) - \Lambda + \eta/2 + \lim_{\epsilon \to 0^+} \epsilon \log  \sum_{m=0}^{\infty} A_{\epsilon,m}.
\label{eq:summation epsilon complicated}
\end{align}
Observe that 
\[
\lim_{\epsilon \to 0^+}\sup_{m\geq 0} A_{\epsilon, m+1} / A_{\epsilon,m} = 0.
\]
 We thus find (using the formula for the sum of a geometric series to upperbound the infinite summation by $2A_{\epsilon,0}$), 
\begin{align}
 \lim_{\epsilon \to 0^+} \epsilon \log  \sum_{m=0}^{\infty}A_{\epsilon,m} \leq  \lim_{\epsilon \to 0^+} \epsilon \log \big( 2 A_{\epsilon,0} \big) = -\infty.
\end{align}
This implies \eqref{eq: eta 1}. Finally, thanks to \eqref{eq: corollary 4.1}, it holds that
\begin{align}
\lim_{\delta \to 0^+}\lsup{\epsilon}\big| \epsilon \log \mathbb{E}\big[ \tau^N_{\rm first} \chi\big\lbrace -\eta < f_{\epsilon} < \eta \big\rbrace  \big] - V(L) + \Lambda \big| \leq \eta. \label{eq: eta 3}
\end{align}
Taking $\eta \to 0^+$, we can thus conclude from \eqref{eq: decompose tau N first three}, \eqref{eq: eta 2}, \eqref{eq: eta 1} and \eqref{eq: eta 3}
that
\begin{align} \label{eq: other main result restated}
\lim_{\delta \to 0^+}\lim_{\epsilon \to 0^+}\epsilon \log \mathbb{E}\big[ \tau^N_{\rm first} \big] = V(L) - \Lambda .
\end{align}
\end{proof}

\subsubsection{Analysis of the first-hitting time for a single walker}
\label{Proofs First Hitting Single Walker}
Classical work by Freidlin and Wentzell \cite{freidlin1998random} describes how the first-hitting-time can be extracted from the Large Deviation Principle \eqref{FW upper bound}. An excellent exposition of these results can be found in the textbook by Dembo and Zeitouni  \cite[Section 5.7]{dembo2009large}. Dembo and Zeitouni determine the most probable asymptotic hitting time. We extend these results to determine the asymptotic likelihood of the extreme first-hitting time being significantly different from the most probable single-walker hitting time. 

We consider the first hitting time asymptotics for a single walker,
\begin{equation}
\label{eq:sde_gen single walker}
dX_t = f( X_t)dt + \sqrt{2\epsilon} dW_t.
\end{equation}
For any $x \in \mathbb{R}^d$, write $P_x \in \mathcal{P}\big( \mathcal{C}([0,\infty) , \mathbb{R}^d) \big)$ to be the probability law of the solution of the SDE \eqref{eq:sde_gen single walker}, started at $x$. Define, for $\delta \ll 1$,
\begin{align}
\tau &= \inf\big\lbrace t \geq 0: X_t \in  B_{\delta}(L) \big\rbrace .
\end{align}

The main result of this section is the following theorem. (Note that the case $b=V(L)$ has been covered in  \cite[Section 5.7]{dembo2009large})
\begin{theorem}
For any $ 0 \leq b < V(L)$ and any $\zeta > 0$ (one is most interested in the case $\zeta \ll 1$),
\begin{align}
\lim_{\delta \to 0^+}\lim_{\epsilon\to 0^+}\left \vert \epsilon \log \mathbb{P}\big( \epsilon \log \tau \in [b-\zeta, b+\zeta ]\big) - b + V(L)  \right \vert &=0 , \label{eq: single walker hitting time asymptotic}\\
\lim_{\delta \to 0^+}\lim_{\epsilon\to 0^+}\epsilon \log \mathbb{P}\big( \epsilon \log \tau < b-\zeta \big) &< b- V(L),  \label{eq: single walker hitting time asymptotic 2} \\
\lim_{\delta \to 0^+}\lim_{\epsilon\to 0^+}\epsilon \log \mathbb{P}\big( \epsilon \log \tau > b+\zeta \big) &< b- V(L).  \label{eq: single walker hitting time asymptotic 3} 
\end{align}
\end{theorem}
\begin{proof}
Write $H_{\rho} = \lbrace x \in \mathbb{R}^d \; : \; \| x \| \leq \rho \rbrace$ and $S_{\rho} = \lbrace x \in \mathbb{R}^d \; : \; \| x \| =2 \rho \rbrace$. For an integer $p\geq 1$, define the stopping times
\begin{align}
\eta_0 =& 0, \\
\eta_{p} =& \inf\big\lbrace t > \alpha_{p} \; : X(t) \in H_{\rho} \cup B_{\delta}(L) \big\rbrace, \\
\alpha_p =& \inf\big\lbrace t > \eta_{p-1} \; : \; X(t) \in S_{\rho} \big\rbrace, \\
\hat{\eta} =& \inf\big\lbrace t \geq 0 \; : X(t) \in   B_{\delta}(L) \big\rbrace .
\end{align}
Intuitively, one expects that it is much more likely that $X(\eta_p) \in H_{\rho}$ than $X(\eta_p) \in B_{\delta}(L)$. We will show that- writing $j$ to be such that $\eta_j =  \hat{\eta}$, $j$ must scale as $\exp\big( -\epsilon^{-1}(b - V(L) \big)$. Indeed since we are dealing with binary events, in essence this probability can be accurately estimated using the binomial theorem. The next step in the proof is to show that the typical values of $\eta_p - \alpha_p$ and $\alpha_p - \eta_{p-1}$ do not significantly vary as long as $p \leq j$. These two steps will then imply that with very high probability it must be that $\tau_1 \simeq  v \exp\big( -\epsilon^{-1}(b - V(L)) \big)$ where $v$ can be loosely thought of as the expected value of $\alpha_p- \eta_{p-1}$. In fact as long as we obtain upper and lower bounds on $\alpha_p- \eta_{p-1}$, we do not need a precise value for $v$, because to leading order it does not affect the scaling of $\tau$, since $\epsilon \log \tau = \epsilon \log v - b + V(L) \simeq V(L) - b$.

We now make these ideas more concrete. Our first aim is to obtain an accurate lower bound on the asymptotic probability $\epsilon \log \mathbb{P}\big( \epsilon \log (\tau)  \in [b-\zeta , b+ \zeta ] \big)$. To do this we identify a set of events $\mathcal{Y}_{\epsilon}(b,\zeta)$ which (i) imply that $\epsilon \log \tau \simeq b$  and (ii) we will see further below that these events are the most likely `means' that $\epsilon \log \tau \simeq b$.

To this end, define the set of indices $I_{\epsilon} = \big\lbrace j\in \mathbb{Z}^+ \; : \;  \epsilon \log j \in [ b- \zeta / 4, b+ \zeta / 4] \big\rbrace$. We first claim that as long as $V(L) - b - 2 \zeta > 0$, as long as $\epsilon$ is sufficiently small,
\begin{equation} \label{eq: final proof first to show}
\mathcal{Y}_{\epsilon}(b,\zeta) \subseteq \big\lbrace \epsilon \log (\tau) \in ( b-\zeta , b+\zeta )  \big\rbrace.
\end{equation}
where (for the constant $C > 0$ in the statement of Lemma \ref{Lemma first hitting time OU} further below),
\begin{multline}
\mathcal{Y}_{\epsilon}(b,\zeta) = \bigg\lbrace \text{ For some }j\in I_{\epsilon} \text{ it holds that }(i) \; \hat{\eta} = \eta_j , (ii) \sup_{1\leq k\leq j}(\eta_k - \eta_{k-1} ) \leq  2\exp\big( \epsilon^{-1} C \rho^2 \big) ,\\
(iii) \inf_{1\leq k \leq j-1}(\alpha_k - \eta_{k-1} ) \geq T( V(L) + b +2 \zeta , \zeta ) \bigg\rbrace  , 
\end{multline}
and $T( V(L) - b - 2 \zeta , \zeta  ) $ is defined in Lemma \ref{Lemma Bound fluctuations uniformly}. It is assumed that $\rho$ is sufficiently small that
\begin{align} \label{eq C rho squared zeta scaling}
C\rho^2 \leq \zeta / 4.
\end{align}
To see why \eqref{eq: final proof first to show} is true, if the event $\mathcal{Y}_{\epsilon}(b,\zeta)$ holds, then it must be that
\begin{align}
\tau &\leq  \sum_{k \leq j}(\eta_k - \eta_{k-1})\nonumber \\
&\leq  \big| I_{\epsilon} \big| \big\lbrace \exp\big( \epsilon^{-1}C\rho^2 \big) + \zeta \big\rbrace \nonumber \\
&\leq  \exp\big( (b+\zeta)/4 \big)\big\lbrace \exp\big( \epsilon^{-1}\zeta / 4 \big) + \zeta \big\rbrace \nonumber \\
&\leq  \exp\big( \epsilon^{-1}(b+\zeta) \big),
\end{align}
as long as $\epsilon$ is sufficiently small. 


It is immediate from \eqref{eq: final proof first to show} that
\begin{align}
\epsilon \log \mathbb{P}\big( \epsilon \log (\tau) \in ( b-\zeta , b+\zeta ) \big)  \geq \epsilon \log \mathbb{P}\big( \mathcal{Y}_{\epsilon}(b,\zeta) \big).
\end{align}
We now demonstrate that for any $a > 0$, for all small enough $\zeta \ll 1$, there exists $\bar{\epsilon}(\zeta,a)$ such that for all $\epsilon < \bar{\epsilon}(\zeta,a)$,
\begin{align}
\epsilon \log \mathbb{P}\big( \mathcal{Y}_{\epsilon}(b,\zeta) \big) &\in \big[ b - V(L) - a , b - V(L) + a \big] ,\label{eq: final proof second to show}
\end{align}
which would imply that
\begin{align}\label{eq: tempry final stages}
\epsilon \log \mathbb{P}\big( \epsilon \log (\tau)  \in [b-\zeta , b+ \zeta ] \big) \geq b - V(L)-a.
\end{align}
To show \eqref{eq: final proof second to show}, we are going to demonstrate that for all sufficiently small $\zeta$, there exists $\bar{\epsilon}(\zeta)$ such that for all $\epsilon < \bar{\epsilon}(\zeta)$, the following four bounds hold
\begin{gather}
\epsilon \log \mathbb{P}\bigg( \text{ For all }j\in I_{\epsilon}, \; \hat{\eta} \neq \eta_j \bigg) < b - V(L) ,  \label{eq: final proof 1} \\
\epsilon \log \mathbb{P}\bigg( \text{ For some }j\in I_{\epsilon}, \; \hat{\eta} = \eta_j \bigg) \in [ b - V(L) - \zeta/2 , b - V(L) + \zeta/2] , \label{eq: final proof 1 1} \\
\epsilon \log \mathbb{P}\bigg( \text{ For some }j\in I_{\epsilon}, \text{ it holds that }\hat{\eta} = \eta_j \text{ and } \sum_{1\leq k\leq j}(\eta_k - \eta_{k-1} ) \geq 2j \exp\big( \epsilon^{-1}C\rho^2 \big)  \bigg) < b-V(L) - 1  \label{eq: final proof 2}  ,\\
\epsilon \log \mathbb{P}\bigg(  \inf_{k \in I_{\epsilon}}(\alpha_k - \eta_{k-1} ) < T( V(L) - b + 2 \zeta , \rho) \bigg) < b - V(L) -2 \zeta . \label{eq: final proof 3} 
\end{gather}
Indeed \eqref{eq: final proof 1}  - \eqref{eq: final proof 3} suffice for \eqref{eq: final proof second to show} since
\begin{multline}
  \mathbb{P}\big( \mathcal{Y}_{\epsilon}(b,\zeta) \big) = \mathbb{P}\bigg( \text{ For some }j\in I_{\epsilon} \text{ it holds that } \hat{\eta} \in \eta_j \bigg)  \nonumber \\
  \qquad\qquad-\mathbb{P}\bigg( \text{ For some }j\in I_{\epsilon}, \text{ it holds that }\hat{\eta} = \eta_j \text{ and } \sum_{1\leq k\leq j}(\eta_k - \eta_{k-1} ) \geq 2j \exp\big(\epsilon^{-1} C \rho^2 \big)  \bigg) \\
 - \mathbb{P}\bigg( \text{ For some }j\in I_{\epsilon}, \text{ it holds that }\hat{\eta} = \eta_j\text{ and } \sum_{1\leq k\leq j}(\eta_k - \eta_{k-1} ) \leq 2j \exp\big(\epsilon^{-1} C \rho^2 \big) \\ \text{ and }  \inf_{k \in I_{\epsilon}}(\alpha_k - \eta_{k-1} ) < T( V(L) - b + 2 \zeta , \rho) \bigg). \label{eq: expansion probability}
\end{multline}
Each of the terms on the right hand side is of the order of $\exp(-c\epsilon^{-1})$. Thanks to \eqref{eq: final proof 2} and \eqref{eq: final proof 3}, the last two terms of the RHS of \eqref{eq: expansion probability} are each of the order of $\exp\big( - \epsilon^{-1}(b - V(L)  - 2\zeta) \big)$. This probability is negligible in comparison to the other term on the RHS of \eqref{eq: expansion probability}, i.e. $\mathbb{P}\big( \text{ For some }j\in I_{\epsilon} \text{ it holds that } \hat{\eta} \in \eta_j \big)$, which scales as $\exp\big( - \epsilon^{-1}(b-V(L)  \pm  \zeta/2) \big)$.

Lets now demonstrate \eqref{eq: final proof 1} and \eqref{eq: final proof 1 1}. Notice that
\begin{align}
 \mathbb{P}\big( \text{ For some }j\in I_{\epsilon}, \; \hat{\eta} \in \eta_j \big) = \sum_{j\in I_{\epsilon}} \mathbb{P}\big( \hat{\eta} \in \eta_j \big) .
\end{align}
Thanks to Lemma \ref{Lemma Large Deviations Applied}, for any $\tilde{\zeta} > 0$, as long as $\rho$ is sufficiently small, for any $j \geq 1$, the conditional probability is such that
\begin{align}
\lsup{\epsilon} \bigg| \epsilon \log \mathbb{P}\big( \hat{\eta} = \eta_j \; | \; \hat{\eta} \geq j-1 \big) - V(L) \bigg| \leq \tilde{\zeta}.
\end{align}
Write $p_{\epsilon} = \exp\big( -\epsilon^{-1}(V(L) + \tilde{\zeta}) \big)$ and $q_{\epsilon} = \exp\big( -\epsilon^{-1}(V(L) - \tilde{\zeta}) \big)$. A simple computation then implies that
\begin{align}
 \sum_{j\in I_{\epsilon}} \mathbb{P}\big( \hat{\eta} \in \eta_j \big) &\leq   \sum_{j\in I_{\epsilon}} q_{\epsilon}\big( 1- p_{\epsilon} \big)^{j-1}, \\
  \sum_{j\in I_{\epsilon}} \mathbb{P}\big( \hat{\eta} \in \eta_j \big) &\geq   \sum_{j\in I_{\epsilon}} p_{\epsilon}\big( 1- q_{\epsilon} \big)^{j-1}.
\end{align}
Now $\epsilon \log \big| I_{\epsilon} \big| \in [b - \zeta / 4 , b + \zeta/4]$, and
\begin{align}
\big( 1- p_{\epsilon} \big)^{j-1} &\leq 1- (j-1) p_{\epsilon} + O\big( j^2 p_{\epsilon}^2 \big), \\
\big( 1- q_{\epsilon} \big)^{j-1} &\geq 1 - (j-1) q_{\epsilon} .
\end{align}
Furthermore for $j \in I_{\epsilon}$, $\epsilon \log (j p_{\epsilon}) \simeq b - V(L) < 0$, and so terms of order $j^2 p_{\epsilon}^2$ are negligible. We can thus infer that for small enough $\epsilon$, and as long as $\tilde{\zeta}$ is sufficiently small
\begin{align}
 \epsilon \log \sum_{j\in I_{\epsilon}} \mathbb{P}\big( \hat{\eta} \in \eta_j \big) \in [b-\zeta/2 , b+\zeta/2].
\end{align}
and therefore \eqref{eq: final proof 1 1} holds. The demonstration of \eqref{eq: final proof 1}  is similar.

We next demonstrate \eqref{eq: final proof 3}. 
\begin{multline}
  \mathbb{P}\bigg( \text{ For some }j\in I_{\epsilon}, \text{ it holds that }\hat{\eta} = \eta_j \text{ and } \sum_{1\leq k\leq j}(\eta_k - \eta_{k-1} ) \geq 2j \exp\big( \epsilon^{-1} C\rho^2 \big)  \bigg)\\
  \leq   \mathbb{P}\bigg( \text{ For some }j\in I_{\epsilon}, \text{ it holds that }\hat{\eta} = \eta_j \text{ and } \sum_{1\leq k\leq j}(\alpha_k - \eta_{k-1} ) \geq j\exp\big( \epsilon^{-1} C\rho^2 \big)  \bigg)\\
+ \mathbb{P}\bigg( \text{ For some }j\in I_{\epsilon}, \text{ it holds that }\hat{\eta} = \eta_j \text{ and }\sum_{1\leq k\leq j}(\eta_k - \alpha_{k} ) \geq j\exp\big( \epsilon^{-1} C\rho^2 \big)    \bigg).
\end{multline}
Employing a union-of-events bound,
\begin{equation}
\begin{split}
  \mathbb{P}\bigg( \text{ For some }j\in I_{\epsilon}, \text{ it holds that }\hat{\eta} = \eta_j \text{ and }\sum_{1\leq k\leq j}(\alpha_k - \eta_{k-1} ) \geq j\exp\big( \epsilon^{-1} C\rho^2 \big)    \bigg)  \\ 
  \leq \big| I_{\epsilon} \big|  \sup_{x\in H_{\rho}} P_x\bigg( \big( \epsilon \log \inf\lbrace t\geq 0 \; : \;  X(t) \in S_{\rho} \rbrace \big) \geq  C \rho^2  \bigg) .
\end{split}
\end{equation}
Since $\epsilon \log | I_{\epsilon} |$ is uniformly upperbounded as $\epsilon \to 0^+$, we thus obtain from Lemma \ref{Lemma first hitting time OU} (further below) that
\begin{multline}
\lim_{\epsilon \to 0^+} \epsilon \log \mathbb{P}\bigg( \text{ For some }j\in I_{\epsilon}, \text{ it holds that }\hat{\eta} = \eta_j \text{ and } \sum_{1\leq k\leq j}(\alpha_k - \eta_{k-1} ) \geq (j-1)C\rho^2  \bigg) = -\infty.
\end{multline}
 Lemma \ref{Lemma first hitting time OU} also implies that
\begin{align}
\lim_{\epsilon \to 0^+} \epsilon \log  \mathbb{P}\bigg( \text{ For some }j\in I_{\epsilon}, \text{ it holds that }\hat{\eta} = \eta_j \text{ and }\sum_{1\leq k\leq j}(\eta_k - \alpha_{k} ) \geq j\exp\big( \epsilon^{-1} C\rho^2 \big)    \bigg).
\end{align}
We may therefore conclude that \eqref{eq: final proof 3} holds.

Since we have established \eqref{eq: final proof 1}-\eqref{eq: final proof 3} , we may conclude that \eqref{eq: final proof second to show} is true. This asymptotic estimate sets the benchmark against which the probabilities of other events will be compared.

Now define the event
\begin{multline}
\mathcal{H}_{\epsilon,\zeta} = \bigg\lbrace \text{ For any }j\in I_{\epsilon}\text{ such that }\hat{\eta} = \eta_j \text{, it holds that } \sum_{1\leq k\leq j}(\eta_k - \eta_{k-1} ) \leq 2j \exp\big( \epsilon^{-1}C\rho^2 \big) \text{ and }\\
 \inf_{k \in I_{\epsilon}}(\alpha_k - \eta_{k-1} ) \geq T( V(L) - b + 2 \zeta , \rho)\bigg\rbrace     .
\end{multline}
We have seen from \eqref{eq: final proof 2} and \eqref{eq: final proof 3} that
\begin{equation}
\lsup{\epsilon} \epsilon \log \mathbb{P}\big(  \mathcal{H}_{\epsilon,\zeta}^c \big) \leq b - V(L) - 2\zeta,   
\end{equation}
which means that the probability of $\mathcal{H}_{\epsilon,\zeta}^c$ is negligible compared to the other events. Thus to finish the Theorem, it suffices to prove that as long as $\rho$ is small enough,
\begin{align}
\lim_{\delta \to 0^+}\lim_{\zeta \to 0^+}\lim_{\epsilon\to 0^+}\left \vert \epsilon \log \mathbb{P}\big( \epsilon \log \tau \in [b-\zeta, b+\zeta ], \mathcal{H}_{\epsilon,\zeta} \big) - b + V(L)  \right \vert &=0,  \label{eq: single walker hitting time asymptotic last}\\
\lim_{\delta \to 0^+}\lim_{\epsilon\to 0^+}\epsilon \log \mathbb{P}\big( \epsilon \log \tau < b-\zeta , \mathcal{H}_{\epsilon,\zeta}\big) &< b- V(L),  \label{eq: single walker hitting time asymptotic 2 last} \\
\lim_{\delta \to 0^+}\lim_{\epsilon\to 0^+}\epsilon \log \mathbb{P}\big( \epsilon \log \tau > b+\zeta , \mathcal{H}_{\epsilon,\zeta} \big) &< b- V(L).  \label{eq: single walker hitting time asymptotic 3 last} 
\end{align}
Starting with \eqref{eq: single walker hitting time asymptotic 2 last}, as long as $2\exp(\epsilon^{-1} C \rho^2) \leq \exp(\zeta / 2)$, it follows that
\begin{align}
\mathbb{P}\big( \epsilon \log \tau < b-\zeta , \mathcal{H}_{\epsilon,\zeta}\big) &\leq \mathbb{P}\big( \sum_{j: \eta_j \leq \hat{\eta}}1 \leq \exp\big( \epsilon^{-1}(b - \zeta/4) \big) \big) \nonumber \\
&\leq   \sum_{j: \epsilon \log j \leq b-\zeta/4} q_{\epsilon}\big( 1- p_{\epsilon} \big)^{j-1}\nonumber \\
&\leq \exp\big( \epsilon^{-1}(b-V(L) - \tilde{\zeta}) -\zeta/4\big) + O\big( \exp(\epsilon^{-1}b-2V(L)\epsilon^{-1}) \big),
\end{align}
and as long as $\tilde{\zeta} \ll \zeta$ (which is guaranteed as long as $\delta \to 0^+$), we obtain \eqref{eq: single walker hitting time asymptotic 2 last}. The proofs of \eqref{eq: single walker hitting time asymptotic last} and \eqref{eq: single walker hitting time asymptotic 3 last} are analogous.

\end{proof}

\begin{lemma}\label{Lemma first hitting time OU}
There exists a constant $C > 0$ ($C$ is independent of $\epsilon$ and $\rho$) such that for all sufficiently small $\rho$
\begin{align}
\lim_{\epsilon \to 0^+} \sup_{x\in H_{\rho}} \epsilon \log P_x\bigg( \epsilon \log \inf\lbrace t\geq 0 \; : \;  X(t) \in S_{\rho} \rbrace \geq  C \rho^2  \bigg) = -\infty \label{eq: OU lemma to prove 1},\\
\lim_{\epsilon \to 0^+} \sup_{x\in S_{\rho}} \epsilon \log P_x\bigg( \epsilon \log \inf\lbrace t\geq 0 \; : \;  X(t) \in H_{\rho} \rbrace \geq  C \rho^2  \bigg) = -\infty\label{eq: OU lemma to prove 2}.
\end{align}
\end{lemma}
\begin{proof}
We provide a proof for \eqref{eq: OU lemma to prove 1}. The proof of \eqref{eq: OU lemma to prove 2} is similar, and also easier, because the deterministic dynamics is attracted towards the origin (by assumption).

Applying Ito's Lemma to the SDE \eqref{eq:sde_gen}, it follows that
\begin{equation}
\label{eq:sde_gen restated in final proof}
d \| X_t \|^2 = \big\lbrace 2\langle  X_t , f( X_t) \rangle +  \epsilon d \big\rbrace dt + 2\sqrt{\epsilon}\langle X_t ,  dW_t \rangle .
\end{equation}
Since the map $u \to \sqrt{u}$ is twice continuously differentiable for $u > 0$, in the case that $\| X_t \| > 0$, another application of Ito's Lemma yields that
\begin{align} \label{eq: norm X t SDE}
d \| X_t \| = \big\lbrace \| X_t \|^{-1} \langle  X_t , f( X_t) \rangle +\frac{ \epsilon d}{2} \| X_t \|^{-1} - \frac{\epsilon}{2}\| X_t \|^{1/2} \big\rbrace dt
+\sqrt{2\epsilon} \| X_t \|^{-1} \langle  X_t , dW_t \rangle.
\end{align}
Since $f(0) = 0$ $0$ is a stabled fixed point and $f$ is Lipschitz, there must exist constants $a, \hat{\rho} > 0$ such that for all $x \in \mathbb{R}^d$ such that $\| x \| \leq  2\hat{\rho}$,
\begin{align}
  \langle  x , f( x) \rangle &\geq  -  2a \| x\|^2 ,\\
   \langle  x , f( x) \rangle &\leq  -  a \| x\|^2 .
\end{align}
It is not hard to check that for any $T > 0$,  $\inf_{t \in [0,T]} \| X_t \| > 0$ with unit probability, since the term $\frac{  \epsilon d}{2} \| X_t \|^{-1}$ dominates the drift in the limit as $ \| X_t \|$ goes to $0$. Hence the SDE in \eqref{eq: norm X t SDE} is always well-defined.

In the case that $\| X_t \| \in [\rho , 2\rho ]$,
\begin{align}
d \| X_t \| \leq \big\lbrace -a\| X_t \|   +\frac{  \epsilon d}{2 \rho} - \frac{\epsilon}{2} \sqrt{\rho} \big\rbrace dt
+\sqrt{2\epsilon} \| X_t \|^{-1} \langle  X_t , dW_t \rangle . \label{eq: differential inequality}
\end{align}
Write
\begin{align}
w(t) = \int_0^t  \| X_s \|^{-1} \langle  X_s , dW_s \rangle .
\end{align}
Now define the sequence of stopping times,
\begin{align}
\beta_0 &= 0 ,\\
\beta_j &= \inf \big\lbrace t\geq \beta_{j-1} + a^{-1} \; : \; \| X_t \| \geq \rho \big\rbrace .
\end{align}
For a positive integer $j \geq 0$, define the event
\begin{align}
\mathcal{Z}_{\epsilon,\rho,j} = \bigg\lbrace \sup_{t \in [\beta_j, \beta_{j+1}]} \big| w_t - w_{ \beta_j} \big| \leq   (2\epsilon)^{-1/2} \rho / 2\bigg\rbrace .
\end{align}
Define $q \in \mathbb{Z}^+$ to be the smallest positive integer such that 
\begin{align}
 \sup_{t \in [\beta_q, \beta_{q+1}]} \big| w_t - w_{ \beta_q} \big| >   (2\epsilon)^{-1/2} \rho / 2.
\end{align}
Note that $q$ is a random variable. Thanks to Lemma \ref{eq: bound on Z epsilon rho k}, it suffices that we show that there exists a constant $\bar{C} > 0$ such that
\begin{align}
\lim_{\rho \to 0^+}\lim_{\epsilon \to 0^+}\sup_{x\in H_{\rho}} \epsilon \log P_x \bigg( \epsilon \log q \geq \bar{C}\rho^2 \bigg) = - \infty.
\end{align}
Write $C$ to be the constant in Lemma \ref{eq: bound rare fluctuations of stochastic integral}, and write
\begin{align}
p_{\epsilon,\rho} = \exp\big( - C \rho^2 a \epsilon^{-1} \big).
\end{align}
We find that, writing $I_{\epsilon,\rho}$ to consist of all positive integers $r$ such that $\epsilon \log r \geq \bar{C}\rho^2$,
\begin{equation}
P_x \bigg(  \epsilon \log q \geq \bar{C}\rho^2 \bigg) \leq \sum_{ k \in I_{\epsilon,\rho}} (1- p_{\epsilon})^k \leq  \sum_{ k \in I_{\epsilon,\rho}} \exp\big(-k p_{\epsilon,\rho} \big)
\end{equation}
Hence
\begin{align}
\lim_{\epsilon \to 0^+}\epsilon \log P_x \bigg(  \epsilon \log q \geq \bar{C}\rho^2 \bigg)   \leq \lim_{\epsilon \to 0^+}\epsilon \log \bigg(\sum_{ k \in I_{\epsilon,\rho}} \exp\big(-k p_{\epsilon,\rho} \big) \bigg).
\end{align}
To bound the above summation, we first look at the ratio of successive terms,
\begin{align}
\lim_{\epsilon to 0^+}\sup_{k \in I_{\epsilon,\rho}} \exp\big(-(k+1) p_{\epsilon,\rho} \big) / \exp\big(-k p_{\epsilon,\rho} \big) = 0.
\end{align}
This means that, upperbounding the summation by the sum of a geometric sequence $1+c+c^2+c^3 +\ldots\leq 2$ (for any small enough $c \ll 1$), as long as $\epsilon$ is sufficiently small,
\begin{align}
\sum_{ k \in I_{\epsilon,\rho}} \exp\big(-k p_{\epsilon,\rho} \big) \leq 2\exp\big( -\exp(\epsilon^{-1} \bar{C} \rho^2 )p_{\epsilon,\rho} \big).
\end{align}
As long as $\bar{C}  > Ca$, the lemma must hold upon taking $\epsilon \to 0^+$.

\end{proof}

\begin{lemma} \label{eq: bound rare fluctuations of stochastic integral}
There exists a constant $C > 0$ such that for all small enough $\rho >0$ and all $\epsilon > 0$,
\begin{align}
\epsilon \log \sup_{x\in H_{\rho}} P_x \big( \mathcal{Z}_{\epsilon,\rho,1}^c \big) \leq - C \rho^2 a.
\end{align}
\end{lemma}
\begin{proof}
One checks that the quadratic variation of $w(t)$ is identically one. Thus - since the quadratic variation uniquely determines the distribution of $w$ \cite{karatzas1998brownian} - it must be that $w(t)$ is identical to Brownian Motion (in probability law). We therefore obtain using Doob's Submartingale Inequality that, for a constant $v > 0$,
\begin{equation}
\begin{split}
 P_x\bigg(  \sup_{t \in [0, a^{-1}]} \big( w_t - w_{ 0} \big) >   (2\epsilon)^{-1/2} \rho / 2 \bigg) &\leq \mathbb{E}^{P_x}\bigg[ \exp\bigg( v \epsilon^{-1/2} \big( w_{a^{-1}} - w_{ 0} \big) - \frac{v \rho}{2}2^{-1/2} \epsilon^{-1} \bigg) \bigg] \\
 &= \exp\bigg( v^2 \epsilon^{-1} a^{-1} / 2  - \frac{v \rho}{2}2^{-1/2} \epsilon^{-1} \bigg),
\end{split}
\end{equation}
since $w_{a^{-1}} - w_0$ is centered Gaussian, with variance $a^{-1}$. To optimize this inequality we choose $v = a\rho 2^{-3/2}$, and we find that there exists a constant $c > 0$ such that
 \begin{align}
 P_x\bigg(  \sup_{t \in [0, a^{-1}]} \big( w_t - w_{ 0} \big) >   (2\epsilon)^{-1/2} \rho / 2 \bigg) \leq \exp\bigg( - \frac{c\rho^2 a}{\epsilon} \bigg).
\end{align}
We similarly obtain that 
 \begin{align}
 P_x\bigg(  \sup_{t \in [0, a^{-1}]} \big( w_t - w_{ 0} \big) <  - (2\epsilon)^{-1/2} \rho / 2 \bigg) \leq \exp\bigg( - \frac{c\rho^2 a}{\epsilon} \bigg).
\end{align}
These two results imply the Lemma.
\end{proof}

\begin{lemma} \label{eq: bound on Z epsilon rho k}
There exists $\bar{\epsilon}_{\rho}>0$ such that for all $\epsilon < \bar{\epsilon}_{\rho}$ the following is true. If for some $j\in \mathbb{Z}^+$ all of the events $\lbrace \mathcal{Z}_{\epsilon,\rho,k} \rbrace_{k \leq j}$ hold, as long as $X_0 \in H_{\rho}$, it holds that
\begin{align}
\sup_{t \leq (j+1)a^{-1}} \| X_t \| \leq 2\rho.
\end{align}
\end{lemma}
\begin{proof}
Define the stopping time 
\begin{align}
\hat{\tau} = \inf\big\lbrace t\geq 0 \; : \; \| X_t \| = 2\rho \big\rbrace ,
\end{align}
and define the (random) index $\ell \in \mathbb{Z}^+$ to be such that
\begin{equation}\label{eq: ell definition}
\beta_{\ell} = \sup\big\lbrace \beta_k \leq \hat{\tau} \; : \; \| X_{\beta_k} \| = \rho \big\rbrace .
\end{equation}
Suppose for a contradiction that $\ell \leq j$.
Thanks to the inequality in \eqref{eq: differential inequality}, and employing the bound in the definition of $\mathcal{Z}_{\epsilon,\rho,\ell} $, for all $t\in  [ \beta_{\ell}, \beta_{\ell+1}]$,
\begin{align} \label{eq: X ell bound}
 \| X_t \| - \| X_{\alpha_{\ell}}  \| \leq   - a(t - \alpha_{\ell}) \rho +(t-\alpha_{\ell}) \bigg( \frac{ \epsilon d}{2 \rho} - \frac{\epsilon}{2} \sqrt{\rho} \bigg) +  \frac{ \rho}{2}
\end{align}
As long as $\bar{\epsilon}_{\rho}$ is sufficiently small, since $\epsilon \leq \bar{\epsilon}_{\rho}$, it must be that
$
-a +  \frac{\epsilon d}{2 \rho} - \frac{\epsilon}{2} \sqrt{\rho} \leq -a/2
$.  
It then follows from \eqref{eq: X ell bound} that 
\begin{align}
\sup_{t\in  [ \beta_{\ell}, \beta_{\ell+1}] } \big\lbrace  \| X_t \| - \| X_{\beta_{\ell}} \| \big\rbrace \leq \rho / 2,
\end{align}
which implies that 
\begin{align}
\sup_{t\in  [ \beta_{\ell}, \beta_{\ell+1}] }  \| X_t \| < 2\rho \label{eq: last temporary alpha ell 0} \\
\| X_{\beta_{\ell+1}} \| \leq \rho , \label{eq: last temporary alpha ell}
\end{align}
where in \eqref{eq: last temporary alpha ell} we have substituted $\beta_{\ell+1} - \beta_{\ell} \geq a^{-1} $. Now \eqref{eq: last temporary alpha ell 0} implies that $\hat{\tau} > \beta_{\ell+1}$. However we would then find that \eqref{eq: last temporary alpha ell} contradicts the definition of $\ell$ in \eqref{eq: ell definition}.

\end{proof}

\begin{lemma} \label{Lemma Large Deviations Applied}
\begin{align}
\lim_{\delta \to 0^+}\lim_{\rho \to 0^+} \lim_{\epsilon \to 0^+}\sup_{x\in S_{\rho}} \epsilon \log P_x\big( \inf\lbrace t  \; : \; X(t) \in B_{\delta}(L) \rbrace < \inf\lbrace t > 0 \; : \; X(t) \in H_{\rho}  \rbrace  \big) = - V(L), \\
\lim_{\rho \to 0^+} \lim_{\delta \to 0^+}\lim_{\epsilon \to 0^+}\sup_{x\in S_{\rho}} \epsilon \log P_x\big( \inf\lbrace t  \; : \; X(t) \in B_{\delta}(L) \rbrace < \inf\lbrace t > 0 \; : \; X(t) \in H_{\rho}  \rbrace  \big) = - V(L) .
\end{align}
\end{lemma}
\begin{proof}
Thanks to the Large Deviation Principle in \eqref{FW upper bound}-\eqref{FW lower bound}, as long as $\delta,\rho$ are sufficiently small, for any $x\in S_{\rho}$,
\begin{equation}
\begin{split}
 &\lim_{\epsilon \to 0^+}\epsilon \log P_x\big( \inf\lbrace t  \; : \; X(t) \in B_{\delta}(L) \rbrace < \inf\lbrace t > 0 \; : \; X(t) \in H_{\rho}  \rbrace \big) \\ 
  &\qquad= - \lim_{T\to\infty} \inf_{y \in H_1 : y(0) =x  ,  y(T) \in B_{\delta}(L)} I_T(y) \\
 &\qquad := -V_{\delta,\rho,x}.
\end{split}
\end{equation}
Since the rate function $I_t$ is lower semi-continuous, it must be that
\begin{align}
\lim_{\delta \to 0^+} \lim_{\rho \to 0^+} V_{\delta,\rho,x}= \lim_{\rho \to 0^+} \lim_{\delta \to 0^+} V_{\delta,\rho,x} = V(L).
\end{align}
\end{proof}

The following Lemma is proved in Dembo and Zeitouni  \cite[Section 5.7]{dembo2009large}.
\begin{lemma} \label{Lemma Bound fluctuations uniformly}
For every $\rho > 0$ and every $c > 0$, there exists a constant $T(c,\rho) < \infty$ such that
\begin{align}
\lsup{\epsilon} \epsilon \log \sup_{x\in \mathbb{R}^d \; : \| x \| \leq 3\rho } P_x\bigg( \sup_{t\in [0,T(c,\rho)]} \| X_t - x \| \geq \rho \bigg) < - c.
\end{align}
\end{lemma}

\subsection{Proof in the case of a non-diverging first-hitting-time (Regime 2)}

We now consider the case that $\Lambda > V(L)$. In this case, the average first-hitting-time does not diverge as $\epsilon \to 0^+$.  Let $T_{\Lambda} =\Psi_L(\Lambda)$ (as defined in \eqref{eq: Psi x c}), which means that
\begin{align}
\Lambda = V(L, T_{\Lambda}).  \label{eq: expected first hitting time}  
\end{align}
As in Regime 1, we fix the number of walkers to be $N = \lfloor \exp( \epsilon^{-1} \Lambda ) \rfloor$ and we (i) asymptote $\epsilon \to 0^+$ before (ii) asymptoting $\delta \to 0^+$.

\begin{lemma} \label{Lemma Regime 2}
For any $\zeta > 0$,
\begin{align}
\lim_{\delta \to 0^+}\lsup{\epsilon}\epsilon \log \mathbb{P}\big( | \tau^N_{\rm first} - T_{\Lambda} | > \zeta \big) < 0 ,\label{eq: epsilon log bound again}\\
\lim_{\delta \to 0^+}\lim_{\epsilon \to 0^+}\mathbb{E}\big[  \tau^N_{\rm first} \big] = T_{\Lambda} .\label{eq: anoth tau N first expectation}
\end{align}
\end{lemma}
\begin{proof}
The proof of \eqref{eq: anoth tau N first expectation} is very similar to the proof of \eqref{eq: other main result}, and is neglected.


The Large Deviations estimates in \eqref{FW upper bound}-\eqref{FW lower bound} imply that for any $a > 0$ (and in particular for $a \ll 1$),
\begin{align}
\lim_{\delta \to 0^+}\lim_{\zeta \to 0}\lim_{\epsilon\to 0}\bigg|\epsilon \log \mathbb{P}\big( \tau_1 \in [T_{\Lambda} - \zeta, T_{\Lambda} + \zeta ]\big) + \Lambda  \bigg| &=0 , \label{eq: single walker hitting time asymptotic repeated 2}\\
\lim_{\delta \to 0^+}\lim_{\epsilon\to 0}\epsilon \log \mathbb{P}\big(  \tau_1 < T_{\Lambda} - a \big)&< - \Lambda \label{eq: single walker hitting time asymptotic repeated 2'}, \\
\lim_{\delta \to 0^+}\lim_{\epsilon\to 0}\epsilon \log \mathbb{P}\big( \tau_1 > T_{\Lambda} + a \big) &< -\Lambda .\label{eq: single walker hitting time asymptotic repeated 3} 
\end{align}
In order that \eqref{eq: epsilon log bound again} holds, it suffices that we demonstrate that for small enough $\zeta > 0$,
\begin{align}
\lim_{\delta \to 0^+}\lsup{\epsilon}\epsilon \log \mathbb{P}\big(  \tau^N_{\rm first} - T_{\Lambda}  > \zeta \big) < 0, \label{eq: epsilon log bound again 2} \\
\lim_{\delta \to 0^+}\lsup{\epsilon}\epsilon \log \mathbb{P}\big(  \tau^N_{\rm first} - T_{\Lambda}  < -\zeta \big) < 0. \label{eq: epsilon log bound again 3} 
\end{align}
To prove \eqref{eq: epsilon log bound again 3}, by a union-of-events bound,
\begin{equation}
\begin{split}
\mathbb{P}\big(  \tau^N_{\rm first} - T_{\Lambda}  < -\zeta \big) &\leq \sum_{j=1}^N \mathbb{P}\big(  \tilde{\tau}^j - T_{\Lambda}  < -\zeta \big) \\
&\leq N \mathbb{P}\big(  \tau^1 - T_{\Lambda}  < -\zeta \big) \\
&\leq N \exp\big(-\epsilon^{-1}\Lambda - \epsilon^{-1}\delta_{\zeta} \big).
\end{split}
\end{equation}
for some $\delta_{\zeta} > 0$, thanks to \eqref{eq: single walker hitting time asymptotic repeated 2'}. Since $N = \lfloor \exp(\epsilon^{-1}\Lambda) \rfloor$, we may conclude \eqref{eq: epsilon log bound again 3}. For \eqref{eq: epsilon log bound again 2}, since the walkers are independent,
\begin{equation}
\begin{split}
\mathbb{P}\big(  \tau^N_{\rm first} - T_{\Lambda}  > \zeta \big) &= \prod_{j=1}^N \mathbb{P}\big(  \tau^j - T_{\Lambda}  > \zeta \big) \\
& \leq  \big( \exp(-\epsilon^{-1}\Lambda - \epsilon^{-1}\delta_{\zeta}) \big)^N ,
\end{split}
\end{equation}
for some $\delta_{\zeta} > 0$, thanks to \eqref{eq: single walker hitting time asymptotic repeated 3}. Since $N\to \infty$ as $\epsilon \to 0$, it is immediate that 
\begin{align}
\lim_{\epsilon\to 0^+}\epsilon \log \mathbb{P}\big(  \tau^N_{\rm first} - T_{\Lambda}  > \zeta \big) = -\infty .
\end{align}
\end{proof}

\subsection{Regime 3: Labor Surplus}

In this regime, we assume that the number of walkers $N$ (the dependence of $N$ on $\epsilon$ is neglected from the notation) is such that
\begin{align}
\lim_{\epsilon \to 0^+}\epsilon \log N = \infty.
\end{align}
Aside from this assumption, $N$ can be arbitrarily large (the larger $N$ is, the closer the limiting trajectory is to a straight line). This case has been extensively analyzed by Lawley and co-workers, see for example \cite{madrid2020competition,linn2022extreme}.

We find that the trajectory followed by the first-to-arrive walker is asymptotically close to a straight line. To specify this optimal trajectory more precisely, for $T > 0$, define $Z^{(T)} \in \mathcal{C}\big( [0,T], \mathbb{R}^d \big)$ to be such that $Z^{(T)}_0 = 0$ and
\begin{align} \label{eq: Z (T) definition}
\frac{dZ^{(T)}}{dt} = f(Z^{(T)}_t) +  T^{-1}L.
\end{align}
One can check using Gronwall's Inequality that these trajectories becomes closer and closer to a straight line from $0$ to $L$ as $T\to 0^+$, i.e.
\begin{align}
\lim_{T\to 0^+}\sup_{S \leq T}\sup_{0\leq t \leq S}\norm{Z^{(S)}_t - t S^{-1} Z^{(S)}_S} = 0.
\end{align}

Write $q \leq N$ to be the (random) index of the first walker to hit $B_{\delta}(L)$, i.e. 
\begin{align}
\tau_q = \tau^N_{\rm first}.
\end{align}
(Note that with unit probability, no two walkers hit $B_{\delta}(L)$ simultaneously.)
\begin{theorem}
For any $\zeta > 0 $ (typically $\zeta \ll 1$),
\begin{align}\label{eq: regime three first result}
\lim_{\delta \to 0^+}\lim_{\epsilon \to 0^+} \epsilon  \log \mathbb{P}\big(  \tau^N_{\rm first}(1+\zeta)  - (\epsilon \log N )^{-1} \norm{L}^2 / 4  \leq 0 \big) &= -\infty\\
\lim_{\delta \to 0^+}\lim_{\epsilon \to 0^+} \epsilon   \log \mathbb{P}\big(  \tau^N_{\rm first}(1-\zeta)  - (\epsilon \log N )^{-1} \norm{L}^2 / 4  \geq 0 \big) &=-\infty \label{eq: regime three second result}\\
\lim_{\delta \to 0^+}\lim_{\epsilon \to 0^+} \epsilon \log \mathbb{P}\big( \sup_{t\in [0,\tau^N_{\rm first}]} \norm{ X^q_t -  Z_t^{(\tau^N_{\rm first})} } \geq \zeta \big) &= -\infty \label{eq: straight line optimal trajectory}\\
\lim_{\delta \to 0^+}\lim_{\epsilon \to 0^+}  \big|  \epsilon \log N \mathbb{E}\big[   \tau^N_{\rm first}\big] - \norm{L}^2 / 4 \big|  &= 0.\label{expectation in regime three}
\end{align}
\end{theorem}
\begin{proof}
Write $\hat{T}_{\epsilon} = \epsilon^{-1}(\log N)^{-1} \norm{L}^2 / 4$, and notice that
\begin{align}
N\exp\bigg( - \frac{\norm{L}^2}{4 \epsilon \hat{T}_{\epsilon}} \bigg) = 1. \label{eq: N scaling identity regime 3}
\end{align}
Employing a union-of-events bound, 
\begin{align}
\mathbb{P}\big(  \tau^N_{\rm first}(1+\zeta)  - (\epsilon \log N )^{-1} \norm{L}^2 / 4 < 0 \big)
\leq \sum_{j=1}^N \mathbb{P}\big(  \tau_{j}(1+\zeta)  -(\epsilon \log N)^{-1}  \norm{L}^2 / 4 < 0 \big).
\end{align}
Thanks to Lemma \ref{Lemma Regime 3}, for any $\zeta > 0$ there must exist $\alpha_{\zeta} > 0$ such that for all sufficiently small $\epsilon$,
\begin{align}
\sum_{j=1}^N \mathbb{P}\big( (1+\zeta) \tau_{j}   - (\epsilon \log N)^{-1}\norm{L}^2 / 4 < 0\big) &\leq  N\exp\bigg( -(\hat{T}\epsilon)^{-1}  \frac{\norm{L}^2}{4} -(\hat{T}\epsilon)^{-1}\alpha_{\zeta}  \bigg) \\
&= N^{-\alpha_{\zeta}},
\end{align}
making use of \eqref{eq: N scaling identity regime 3}. This implies \eqref{eq: regime three first result}, after one first takes the logarithms of both sides, and then multiplies by $\epsilon$.

For the other inequality,
\begin{align}
 \mathbb{P}\big(  \tau^N_{\rm first}(1-\zeta)  - (\epsilon \log N )^{-1} \norm{L}^2 / 4  \geq 0 \big) &= \prod_{j=1}^N \big( 1-  \mathbb{P}\big(  \tau_j(1-\zeta)  - (\epsilon \log N )^{-1} \norm{L}^2 / 4  < 0 \big) \big) \nonumber \\
 & \leq  \prod_{j=1}^N \big\lbrace 1 - \exp\big( - \eta_{\zeta} (\epsilon \hat{T}_{\epsilon})^{-1}\big) \big\rbrace , 
\end{align}
for some $\eta_{\zeta} > 0$, thanks to Lemma \ref{Lemma Regime 3}. Now
\[
\prod_{j=1}^N \big\lbrace 1 - \exp\big( - \eta_{\zeta} (\epsilon \hat{\tau})^{-1}\big) \big\rbrace \leq \exp\bigg( -N\exp\big( - \eta_{\zeta} (\epsilon \hat{\tau})^{-1}\big)  \bigg),
\]
and one readily obtains \eqref{eq: regime three second result}.

The proof of \eqref{expectation in regime three} is analogous to proofs of the expectation in the previous sections and is neglected. \eqref{eq: straight line optimal trajectory} is proved in Lemma \ref{Lemma optimal traec}.
\end{proof}
 

\subsubsection{Single-Walker Asymptotics}

We reconsider the first hitting time asymptotics for a single walker,
\begin{equation}
\label{eq:sde_gen single walker regime 3}
dX_t = f( X_t)dt + \sqrt{2\epsilon} dW_t.
\end{equation}
As previously, define
\begin{align}
\tau &= \inf\big\lbrace t \geq 0: X_t \in  B_{\delta}(L) \big\rbrace .
\end{align}
Our main result in this section is the following lemma. 
\begin{lemma}\label{Lemma Regime 3}
For any $\zeta >0 $, there exists $\bar{\epsilon}_{\zeta} > 0$ and $\bar{T}_{\zeta,\delta} > 0$ such that (i) of  $\bar{\epsilon}_{\zeta}$ and $(\bar{\epsilon}_{\zeta})_{\zeta > 0}$ decreases to $0$ as $\zeta \to 0^+$ and for any $\eta > 0$,
\begin{align}
 \lim_{\zeta \to 0^+} \lim_{\delta \to 0^+} \sup_{T \leq \bar{T}_{\zeta,\delta}}\sup_{\epsilon \leq \bar{\epsilon}_{\zeta}} T\epsilon \log \mathbb{P}\big(   \tau \in [T(1 - \zeta), T(1+ \zeta)] \big) &= - L^2 / 4 \\
 \lim_{\zeta \to 0^+}\lim_{\delta \to 0^+} \sup_{T \leq \bar{T}_{\zeta,\delta}}\sup_{\epsilon \leq \bar{\epsilon}_{\zeta}} T\epsilon \log \mathbb{P}\big(   \tau < T(1 - \eta)\big) &< -L^2 / 4 \\
  \lim_{\zeta \to 0^+}\lim_{\delta \to 0^+}  \sup_{T \leq \bar{T}_{\zeta,\delta}}\sup_{\epsilon \leq \bar{\epsilon}_{\zeta}} T\epsilon \log \mathbb{P}\big(   \tau > T(1 + \eta)\big) &> -L^2 / 4.
\end{align}
Also,
 \begin{align}
 \lim_{T \to 0^+} T V(L,T) = L^2 / 4 .
\end{align}
\end{lemma}
\begin{proof}
We are going to see that as the time interval becomes asymptotically short, the problem reduces to the first-hitting-time of a Brownian Motion (since the advection term becomes negligible).
Fix some $\eta \ll 1$. Define the event (for constants $K,M$ that are independent of $\epsilon,\zeta$ and $\tau$ and are specified in the course of the proof)
\begin{align}
\mathcal{A}_{\epsilon,\zeta,\eta} &= \bigg\lbrace \inf_{t \in [T(1-\zeta / 2),T(1+\zeta / 2)]} \norm{W_t - L} \leq (2\epsilon)^{-1/2} \eta\bigg\rbrace \\
\mathcal{B}_{\epsilon,\zeta} &= \bigg\lbrace \sup_{t \leq T(1+\zeta)} \norm{W_t} \leq 2\norm{L} (2\epsilon)^{-1/2} \bigg\rbrace \\
\mathcal{D}_{\epsilon,\zeta,\delta} &= \bigg\lbrace \inf_{t \in [T(1-\zeta),T(1+\zeta)]} \norm{ (2\epsilon)^{1/2} W_t - L} \leq \delta +  T(1+\zeta)M \bigg\rbrace \\
\mathcal{E}_{\epsilon,\zeta,\eta} &=\bigg\lbrace \sup_{t\in [0, T(1+\zeta)]}\big|  W_t  -  t / (T(1+\zeta/2)) W_{T(1+\zeta/2)} \big| \leq (2\epsilon)^{-1/2}\eta \bigg\rbrace
\end{align}

We next claim that, as long as $\eta$ is small enough,
\begin{align}
\bigg\lbrace \tau \in [T(1-\zeta),T(1+\zeta)] ,\mathcal{B}_{\epsilon,\zeta} \bigg\rbrace &\subseteq \mathcal{B}_{\epsilon,\zeta} \bigcap \mathcal{D}_{\epsilon, \zeta , \delta} \label{eq: left inclusion} \\
\mathcal{A}_{\epsilon,\eta,\delta}\bigcap \mathcal{E}_{\epsilon,\zeta,\eta} \bigcap \mathcal{B}_{\epsilon, \zeta}  &\subseteq \bigg\lbrace \tau \in [T(1-\zeta),T(1+\zeta)] \bigg\rbrace .\label{eq: right inclusion} 
\end{align}
This above two results are helpful because we then only need to estimate the asymptotic probabilities of the events $\mathcal{A}_{\epsilon, \zeta} ,\mathcal{B}_{\epsilon,\zeta},\mathcal{D}_{\epsilon, \zeta,\delta}$ to get accurate upper and lower bounds on the probability of the event $\tau \in [T(1-\zeta),T(1+\zeta)]$. Note that the event $\mathcal{B}_{\epsilon,\zeta}$ is included because it implies a uniform bound on $\norm{X_t}$ for $t\leq \tau$: the probability of $\mathcal{B}_{\epsilon,\zeta}^c$ will be seen to be negligible compared to the other events.

We start with \eqref{eq: left inclusion}. Suppose that $\tau \in [T(1-\zeta),T(1+\zeta)]$ and the event $\mathcal{B}_{\epsilon,\zeta}$ holds. One easily checks using Gronwall's Inequality that $\mathcal{B}_{\epsilon,\zeta}$ implies the uniform bound
\begin{align}
\sup_{t\leq (1+\zeta)T}\norm{X_t} \leq \exp\big( T(1+\zeta)C_f \big)2\norm{L}.
\end{align}
This means that for all $t\leq (1+\zeta)T$, 
\begin{align}
\norm{f(X_t)} \leq C_f \exp\big( T(1+\zeta)C_f \big)2\norm{L} := M.
\end{align}
Thus 
\begin{align}
\norm{X_\tau - \sqrt{2\epsilon}W_{\tau}} \leq \tau M \leq T(1+\zeta)M,
\end{align}
as long as $\tau \in [T(1-\zeta),T(1+\zeta)]$. An immediate consequence of the definition of $\tau$ is that $\norm{X_{\tau} - L} \leq \delta$: we thus find that the event $\mathcal{D}_{\epsilon,\zeta,\delta}$ must hold. We have thus proved \eqref{eq: left inclusion}.

We next prove \eqref{eq: right inclusion}. Suppose that the following three events $\mathcal{A}_{\epsilon,\eta,\delta}$, $\mathcal{E}_{\epsilon,\zeta,\eta}$  , $ \mathcal{B}_{\epsilon, \zeta}$ all hold. Recalling the definition of $Z^{(T+T\zeta/2)}$ in \eqref{eq: Z (T) definition}, for $t \leq T(1+\zeta/2)$,
\begin{align}
\norm{ X_t - Z^{(1+\zeta/2)}_t} \leq C_f \int_0^t \norm{ X_s - Z^{(1+\zeta/2)}_s} ds + 2\eta.
\end{align}
Gronwall's Inequality thus implies that
\begin{align}\label{eq: gronwall reimge 3}
\sup_{t\leq T(1+\zeta/2)}\norm{ X_t - Z^{(1+\zeta/2)}_t}   \leq 2\eta \exp\big( T(1+\zeta/2)C_f \big).
\end{align}
Define
\begin{align}
\tilde{\tau}_{\zeta,\delta} = \inf\big\lbrace t\geq 0 : Z^{(1+\zeta/2)}_t \in B_{\delta}(L) \big\rbrace .
\end{align}
The definition of $Z^{(T+T\zeta/2)}$ in \eqref{eq: Z (T) definition} implies that for any $\zeta > 0$, as long as $T$ and $\delta$ are sufficiently small, 
\begin{align}
\tilde{\tau}_{\zeta,\delta} \in [T(1-\zeta/2),T(1+\zeta/2)]. \label{eq: identity one}
\end{align}
It follows from \eqref{eq: gronwall reimge 3} that for all small enough $\eta > 0$,
\begin{align}
\big| \tau - \tilde{\tau}_{\zeta,\delta} \big| \leq \zeta/2\label{eq: identity two}
\end{align}
The identities \eqref{eq: identity one} and \eqref{eq: identity two} imply that $\tau \in [T(1-\zeta),T(1+\zeta)]$. We have thus proved \eqref{eq: right inclusion}.

The Lemma now follows from the first-hitting-time estimates in Lemma \ref{Lemma Bound Brownian First Hitting Times}.

\end{proof}

We next state some standard first-hitting-time estimates for Brownian Motion. The proofs are omitted: we refer the reader to \cite{ledoux2013probability,borodin2015handbook}.
\begin{lemma} \label{Lemma Bound Brownian First Hitting Times}
For any $\tilde{\eta} > 0$ and any $\zeta < 1$, as long as $\eta$ is sufficiently small,
\begin{align}
\inf_{\epsilon <1} \inf_{T < 1}\big\lbrace \epsilon T(1+\zeta/2) \log \mathbb{P}\big(\mathcal{A}_{\epsilon,\zeta,\eta}\big) \big\rbrace &\geq - \frac{\norm{L}^2}{4} - \tilde{\eta} \\
\lim_{\delta\to 0^+}\sup_{\epsilon <1} \sup_{T < 1} \big\lbrace \epsilon T(1+\zeta) \log \mathbb{P}\big(\mathcal{D}_{\epsilon,\zeta,\delta}\big) \big\rbrace &\leq - \frac{\norm{L}^2}{4} + \tilde{\eta} \\
\sup_{\epsilon <1} \sup_{T < 1} \big\lbrace \epsilon T(1+\zeta) \log \mathbb{P}\big(\mathcal{B}_{\epsilon,\zeta}^c\big) \big\rbrace &\leq - \frac{\norm{L}^2}{2}
\end{align}
\end{lemma}

\begin{lemma} \label{Lemma E epsilon unlikely}
For any $\eta > 0$,
\begin{align}
\sup_{\epsilon \leq 1} \sup_{T \leq 1}\epsilon T(1+\zeta) \bigg\lbrace  \log \mathbb{P}\big( \norm{W_{T(1+\zeta)} - L} \leq \eta ,\mathcal{E}_{\epsilon,\zeta,\eta}^c\big) 
- \log \mathbb{P}\big( \norm{W_{T(1+\zeta)} - L} \leq \eta \big) \bigg\rbrace< 0.
\end{align}
\end{lemma}
\begin{proof}
Standard Gaussian theory dictates that, conditionally on $W_{T(1+\zeta)}$, $\tilde{W}_t := W_t  -  t / (T(1+\zeta)) W_{T(1+\zeta)}$ is a centered Gaussian process, and for each $t \in [0,T(1+\zeta)]$ the variance is such that
\begin{align}
 \mathbb{E}\big[ \norm{\tilde{W}_t}^2 ]   \leq  \mathbb{E}\big[ \norm{W_t}^2 ] \leq T(1+\zeta).
\end{align}
The Borel-TIS Theorem \cite{adler2009random} implies that the conditioned probability is such that for any $\eta > 0$,
\begin{align}
\sup_{\epsilon \leq 1} \sup_{T \leq 1} \bigg\lbrace \epsilon T(1+\zeta) \log \mathbb{P}_{W_{T(1+\zeta)}}\bigg( \sup_{t\in [0,T(1+\zeta)]} \norm{\tilde{W}_t} \geq (2\epsilon)^{-1/2}\eta  \bigg) \bigg\rbrace := -c,
\end{align}
and $c > 0$. Hence
\begin{align}
\mathbb{P}\big( \norm{W_{T(1+\zeta)} - L} \leq \eta ,\mathcal{E}_{\epsilon,\zeta,\eta}^c\big) 
= &\mathbb{E}\bigg[ \chi\big\lbrace \norm{W_{T(1+\zeta)} - L} \leq \eta \big\rbrace 
\mathbb{P}_{W_{T(1+\zeta)}}\big( \sup_{t \leq T(1+\zeta)} \norm{\tilde{W}_t} \geq (2\epsilon)^{-1/2}\eta  \big)\bigg] \nonumber \\
\leq &\mathbb{P}\big( \norm{W_{T(1+\zeta)} - L} \leq \eta \big) \exp\big( -c\big((1+\zeta)T\epsilon\big)^{-1} \big)
\end{align}
This implies the Lemma.

\end{proof}

Next we prove that the trajectory of the particle is most likely close to a straight line (i.e. we prove \eqref{eq: straight line optimal trajectory}).
\begin{lemma}\label{Lemma optimal traec}
For any $\eta,\zeta > 0$,
\begin{align}
\lim_{\delta \to 0^+}\lim_{\epsilon,T \to 0^+} \epsilon T \log \mathbb{P}\bigg( \tau \in [T(1-\zeta), T(1+\zeta)] \; , \; \sup_{t\in [0,\tau]}\norm{X_t - Z^{(\tau)}_t} \geq \eta \bigg) < 0
\end{align}
\end{lemma}
\begin{proof}

One can prove using Gronwall's Inequality that, as long as $\eta,\kappa$ are sufficiently small,
\begin{equation}
\big\lbrace \tau \in [T(1-\zeta), T(1+\zeta)] \; , \; \sup_{t\in [0,\tau]}\norm{X_t - Z^{(\tau)}_t} \geq \eta \bigg\rbrace \subseteq \mathcal{A}_{\epsilon,\zeta,\eta} \bigcap \mathcal{B}_{\epsilon,\zeta} \bigcap\mathcal{E}^c_{\epsilon,\zeta,\kappa}.
\end{equation}
See Lemma \ref{Lemma Gronwalls}. Next we prove that for any $\kappa,\eta > 0$
\begin{align}
\sup_{\epsilon < 1}\sup_{T < 1} \epsilon T\bigg\lbrace \log \mathbb{P}\big(\mathcal{E}_{\epsilon,\zeta,\kappa}, \mathcal{A}_{\epsilon,\zeta,\eta} , \mathcal{B}_{\epsilon,\zeta}\big) - 
\log \mathbb{P}\big( \mathcal{A}_{\epsilon,\zeta,\eta}, \mathcal{B}_{\epsilon,\zeta}\big) \bigg\rbrace 
< 0 
\end{align}
This is proved in Lemma \ref{Lemma E epsilon unlikely}.
\end{proof}

\begin{lemma}\label{Lemma Gronwalls}
Assume that $T \leq 1$. Let $y,z,v,w \in \mathcal{C}([0,T],\mathbb{R}^d)$ be such that $y_0 = z_0 = w_0 = v_0 = 0$, and also such that
 for all $t\in [0,T]$,
\begin{align}
y_t =& \int_0^t f(y_s) ds + \sqrt{2 \epsilon}w_t, \\
z_t =& \int_0^t f(z_s) ds + \sqrt{2 \epsilon}v_t .
\end{align}
There exists a constant $C > 0$ ($C$ is uniform for all $T\leq 1$ and $\epsilon > 0$) such that
\begin{align}
\sup_{t\in [0,T]}\norm{ y_t - z_t} \leq C \sqrt{2\epsilon}\bigg\lbrace \sup_{t\in [0,T]}\norm{v_t - w_t} \bigg\rbrace.     
\end{align}
\end{lemma}
\begin{proof}
Thanks to the triangle inequality,
\begin{align}
\norm{ y_t - z_t} \leq \int_0^t C_f \norm{ y_s - z_s} ds +  \sqrt{2\epsilon}\sup_{s\in [0,t]}\norm{v_s - w_s}.
\end{align}
Gronwall's Inequality now implies that
\begin{align}
\sup_{t\in [0,T]}\norm{ y_t - z_t} \leq \exp\big(T C_f \big) \sqrt{2\epsilon}\sup_{s\in [0,T]}\norm{v_s - w_s}.
\end{align}
\end{proof}

\section{Example: Ornstein-Uhlenbeck process}
\label{sec:example}
To illustrate our theory, we consider the case where $f(x)=Ax$ is a linear function.
The advantage to this simple model is that we can obtain many exact results and it is generally applicable when viewed as a linearization around a stable equilibrium point of a given nonlinear model.
Consider a 2D Ornstein-Uhlenbeck (OU) process given by
\begin{equation}
\label{eq:sde_ou}
     dX_t = AX_t dt + \sqrt{2\epsilon}dW,\quad X_0 = 0,
\end{equation}
where 
$$ A = \begin{bmatrix}-\mu & -\omega \\ \omega & -\mu \end{bmatrix},$$
and $\mu > 0$.
With this choice of matrix $A$, the dynamics are rotational; that is, the origin is a stable spiral for the deterministic dynamics.
As we show in this section, the rotational dynamics impart a nontrivial spiral shape for the globally optimal MLT.

We consider the single walker rare event 
\begin{equation}
\label{eq: example fpt}
    \tau = \inf\{t>0\mid \Vert X_t - X_0\Vert = \Vert L \Vert \},
\end{equation}
for $L\neq 0$. 
Because the probability law is invariant under a rotation of the axes, the large deviations rate function and MLT are also invariant under a rotation of the axes. 
In particular, suppose that $\tilde{L}$ is such that $\Vert \tilde{L}\Vert = \Vert L\Vert$, and $\tilde{L}$ can be obtained by rotating $L$ by some angle $\alpha$.
Then, for any $T>0$, the MLT $\tilde{\varphi}_{t|T}$ for which $\tilde{\varphi}_{T|T}=\tilde{L}$ can be obtained by taking the MLT $\varphi_{t|T}$ for which $\varphi_{T|T}=L$ and rotating it uniformly by $\alpha$.

The linearity of the OU dynamics facilitates analytic formulae. 
Indeed, given our initial condition, the process can be described as a multivariate normal random variable for any countable set of times.
The eigenvalues of $A$ are $-\mu \pm i \omega$. 

The solution to \eqref{eq:sde_ou} can be written as
\begin{equation}
X_t = \sqrt{2\epsilon}   \int_0^t e^{-\mu (t-s)} R_{t-s} dW_s .
\end{equation}
where
\begin{align}
R_{t} = \begin{bmatrix}
\cos(\omega t) & -\sin(\omega t) \\
\sin(\omega t) & \cos(\omega t) 
\end{bmatrix}.
\end{align}

The Hamiltonian defined by \eqref{eq:ham_ldt} is given by
\begin{equation}
\label{eq:ham_ou}
    \mathcal{H}(x, p) = p^TAx + \Vert p\Vert^2.
\end{equation}
For our problem \eqref{eq:ham_dyn_gen} and \eqref{eq:S_of_t} become
\begin{equation}
    \label{eq:ham_dyn}
    \begin{split}
    \dot{\varphi}_{t|T} &= A\varphi_{t|T} + 2p_{t|T}, \\ 
    \dot{p}_{t|T} &= -A^Tp_{t|T}, \\
    \dot{S}_{t|T} & = \Vert p_{t|T}\Vert^2,
    \end{split}
\end{equation}
 with 
\begin{equation}
\varphi_{0|T}=0, \, \varphi_{T|T} = L,\, S_{0|T} = 0.
\label{eq: boundary condition OU example}
\end{equation}


Rather than solve \eqref{eq:ham_dyn} numerically, we can take advantage of the linear OU dynamics and solve it exactly.
One can verify that the ($T < \infty$) suboptimal MLTs are given by
\begin{align}
\label{eq:MLT_sol}
\varphi_{t|T} &= \frac{e^{\mu t} - e^{-\mu t}}{e^{\mu T} - e^{-\mu T}}R_{t - T}L, \\
p_{t|T} &= \frac{\mu e^{\mu t}}{e^{\mu T} - e^{-\mu T}}R_{t - T}L, \\
\label{eq:cost_S}
S_{t|T} &= \frac{\mu \Vert L\Vert^2}{2}\left(\frac{e^{-2\mu T}}{\left(1 - e^{-2\mu T}\right)^2}\right)( e^{2\mu t} - 1).
\end{align}
In fact, one can further verify (see Appendix \ref{sec:example appendix}) that $\varphi_{t|T}$ is the exact conditional mean $\mathbb{E}[X_t \mid X_T=L, X_0=0]$, which happens to be independent of $\epsilon$ in this example.

To complement the exact formulae for the ($T<\infty$) suboptimal MLTs $\varphi_{t|T}$, we next derive the ($T\to \infty$) optimal MLT $\varphi_{s|\infty}$ as a special case.
Recall that the optimal MLT starts at a critical point of \eqref{eq:ham_dyn}, namely $x = p = 0$. 
Hence, it is necessary to consider the solution converging to the critical point in the limit $t\to -\infty$. 
We can write the optimal MLT in backward time $s<0$ so that $\varphi_{0|\infty} = L$ and $\lim_{s\to-\infty} \varphi_{s|\infty} = 0$. 
We find
\begin{align}
\label{eq:opt_MLT}
    \varphi_{s|\infty} &= e^{\mu s}R_s L, \quad -\infty < s < 0,\\
    p_{s | \infty} &= \mu e^{\mu s}R_s L, \\
    S_{s | \infty} &=  \frac{\mu}{2} e^{2\mu s}\Vert L\Vert^2.
\end{align}
We note that $V(L) = S_{0|\infty} = \frac{\mu}{2}\Vert L\Vert^2,$ where $V$ is the quasi potential \eqref{eq:quasi potential}.

We also note that for an OU process (see Appendix \ref{sec:example appendix} for derivation), the cost function \eqref{eq:cost_fn} and solution to the Hamilton-Jacobi equation \eqref{eq: Hamiltonian PDE} is given by the quadratic form
\begin{align} 
\label{eq: OU cost function}
V(x, t) = \frac{1}{2} x^T  U_{tt}^{-1} x = \frac{\mu}{2}\Vert x\Vert^2 \left(\frac{1}{1 - e^{-2\mu t}}\right).
\end{align}

For our example, it is possible to exactly sample the (random) conditional transition paths $X^{\epsilon}_{t| T}$.
It is well-known that the the OU process over a finite set of times is a multivariate normal random variable \cite{karatzas1998brownian}. Furthermore if one conditions a multivariate normal random variable on the values of a subset of the variables, then one obtains another multivariate normal random variable (with a different mean and covariance matrix) \cite{lindgren2013stationary}. 
Discretize the interval $[0, T]$ into $n+2$ equally space points $t_i$, $i=0,\ldots,n+1$, where $t_0 = 0$ and $t_{n+1} = T$. 
Note that initial and end positions are fixed, which leaves $n$ points along the trajectory that must be sampled.
We then define the $2n$-dimensional random variable representing a sampled conditional path with
\begin{equation}
    X_{1:n|T}^{\epsilon} = \{X_{t_1|T}^{\epsilon},\ldots, X_{t_n|T}^{\epsilon}\}.
\end{equation}
We have that
\begin{equation}
\label{eq:conditional samples paths}
    X_{1:n|T}^{\epsilon}  :\sim \text{MultivariateNormal}
    (\varphi_{1:n|T}, \epsilon U_{1:n,\,1:n|T}),
\end{equation}
where $\varphi_{1:n|T}$ from \eqref{eq:MLT_sol} is the conditional mean.
The covariance $U_{1:n,\,1:n|T}$ can be thought of as a $2n\times 2n$ symmetric matrix.
Let $U_{ts|T} \in \mathbb{R}^{2\times 2}$ be the scaled conditional auto covariance matrix, i.e.
\begin{align}
\label{eq:cacov}
U_{ts|T} :=  \epsilon^{-1} \mathbb{E}\left [ X_t X_s^T \mid X_T = L,\, X_0 = 0 \right ] = U_{ts} - U_{tT}U_{TT}^{-1} U_{Ts},
\end{align}
where the scaled auto covariance matrix is given by, \cite[Ch.\ 4, p.\ 109]{gardiner2009stochastic}
\begin{align}
U_{sr} := 2\epsilon^{-1} \mathbb{E}\left[X_{s} X_{r}^T \right] = \frac{2}{\mu}\left(e^{-\mu\vert s - r\vert} - e^{-\mu\vert s + r\vert}\right)R_{s-r}.
\end{align}

Unfortunately, there is no practical exact representation of the single-walker first passage time density or direct method for sampling it \cite{giorno1986some,borodin2015handbook}. 
In general, sampling the extreme first hitting time is not computationally feasible since $N\gg 1$ and $\epsilon \ll 1$. 
Each sample of the extreme first hitting time requires $N$ single-walker first hitting times, each of which must be sampled using a time stepping method such as Euler-Maruyama. 
Since each hitting time is a rare event, it requires an exponentially large number of time steps, on average. 
Hence, the computation of a single extreme first hitting time requires $O(N\frac{e^{\Lambda/\epsilon}}{\Delta t})$ time steps of length $\Delta t$.
Moreover, it is quite likely that the first passage time, because it is a rare event, is highly sensitive to discretization error.
We note, however, that it may be possible to speed this computation up to some degree using importance sampling techniques.

As an alternative to sampling the extreme first hitting time, we illustrate our results with the rigorously established approximation of the mean first hitting time $\mathbb{E}[\tau^N_{\rm first}] \simeq \Psi_L(\Lambda)$, where $\Psi_L(\Lambda)$ is given by \eqref{eq: Psi x c}. 
We want to sample MLT trajectories conditioned on fixed time interval $0 \leq t \leq  T = \Psi_L(\Lambda)$.  
That is, trajectories that begin at the origin at $t=0$ and end at the position $L$ at time $t=\Psi_L(\Lambda)$.

\begin{figure}[tbp]
    \centering
    \includegraphics[width=8cm]{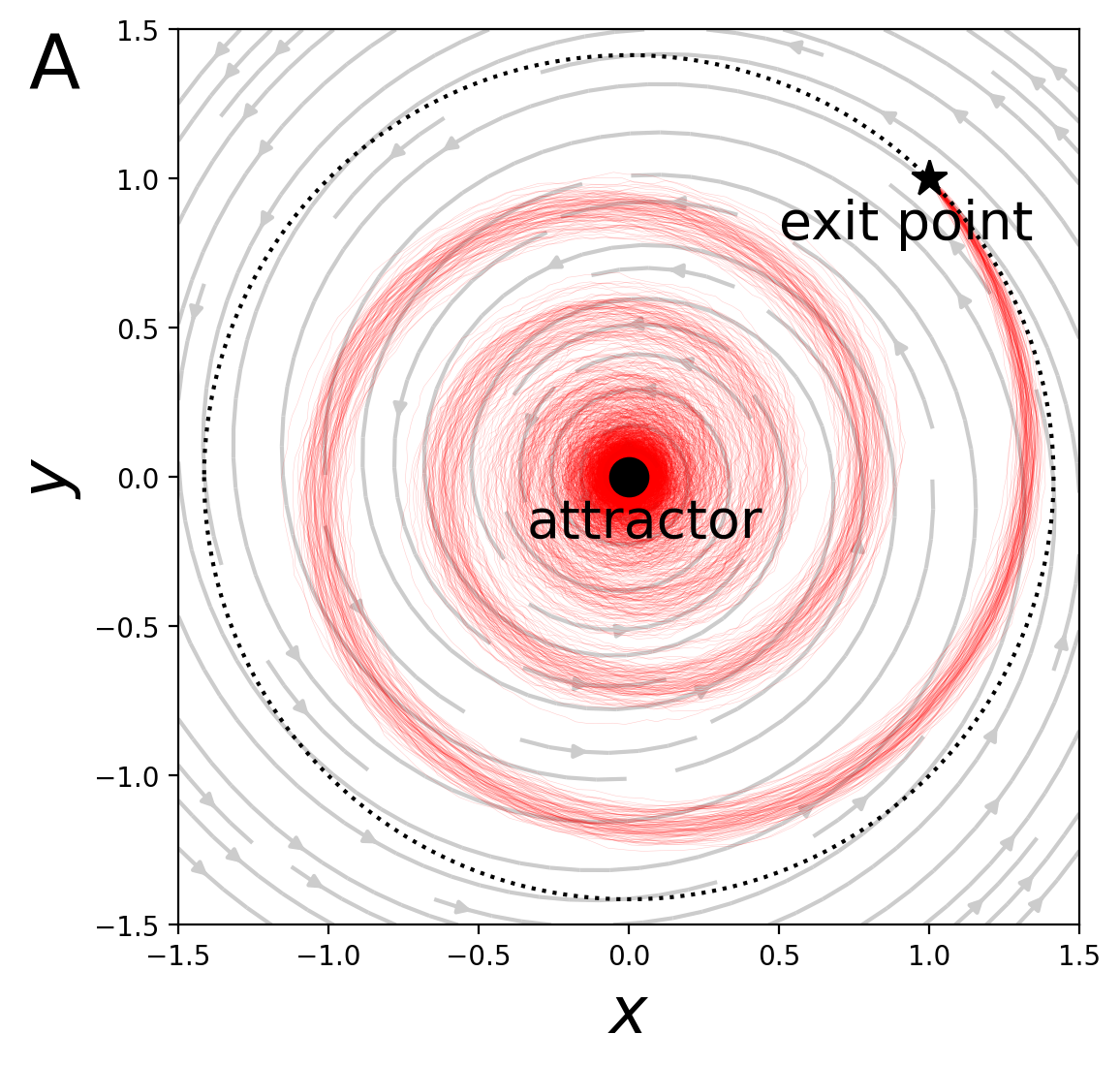}
    \includegraphics[width=13cm]{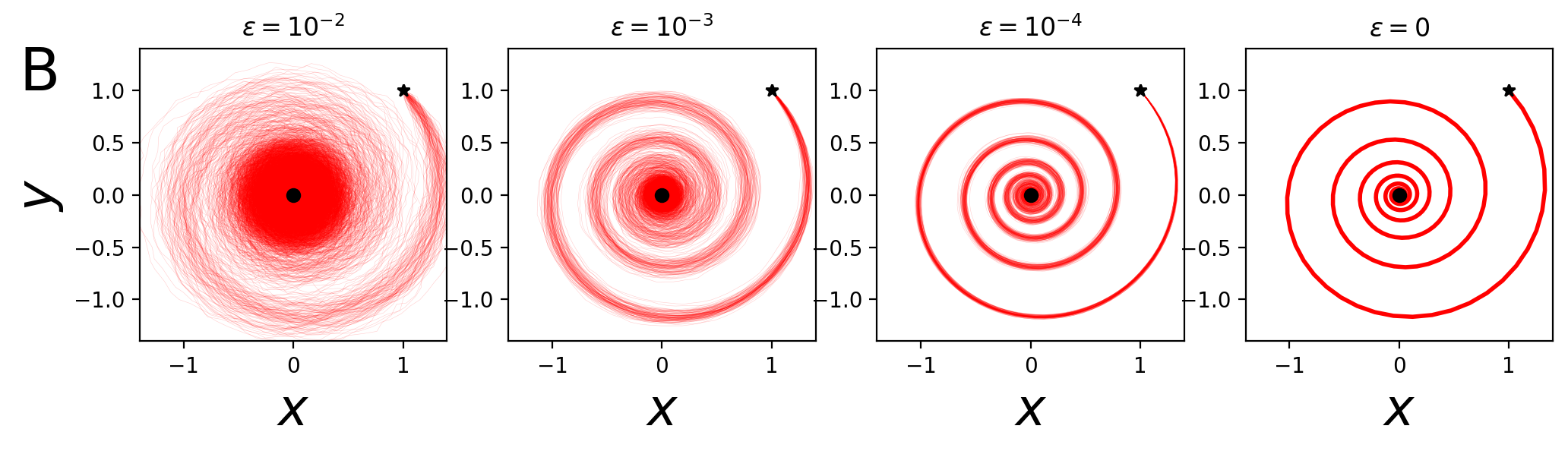}
    \includegraphics[width=8.5cm]{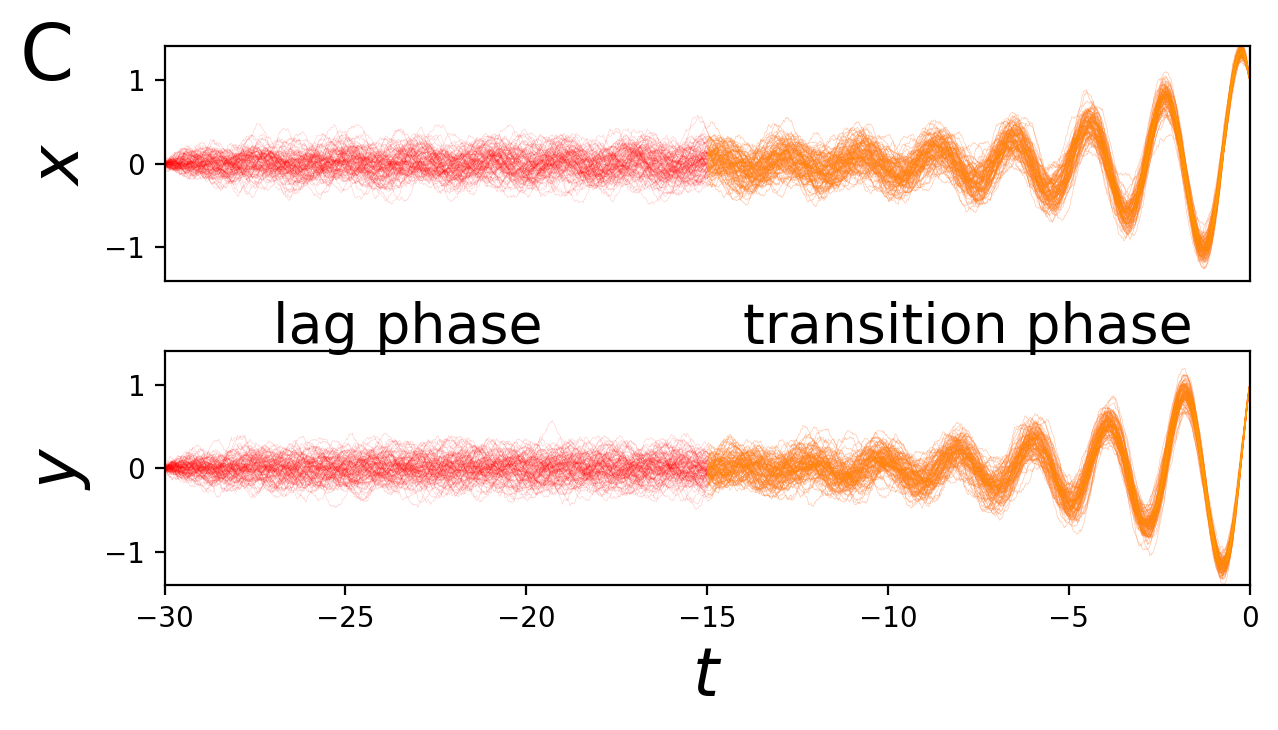}
    \caption{{\bf Single walker rare event.}
    (A) Samples of the optimal escape trajectory \eqref{eq:conditional samples paths} (red) for $\epsilon=10^{-3}$, with streamlines showing the deterministic flows ($\dot{x}=Ax$) for reference. (The optimal trajectory $T=\infty$ is approximated with $T=30$.) (B) The escape trajectories follow a "tube". As the noise magnitude $\epsilon$ goes to zero, the width of the "tube" goes to zero.
    (C) sampled conditional transition paths with $\epsilon=5\times 10^{-3}$ are shown as functions of time. The paths are separated into two segments or phases. The lag phase is shown as red and the transition phase is shown as orange.
    Other parameters are $\mu=0.25$, $\omega=3$.}
    \label{fig:1}
\end{figure}

To illustrate the exit trajectories, we focus on a single exit point $L = (1, 1)$, but we note that any of the points on the circle $\Vert L\Vert = \text{const}$ (dashed line in Fig.~\ref{fig:1}A) are equally likely.
Several samples of conditional transition trajectories \eqref{eq:conditional samples paths} are shown in Fig.~\ref{fig:1}.
As the noise amplitude is decreased, the sampled trajectories condense into a tube like path, the center of which is the MLT \eqref{eq:MLT_sol}.
Several sample conditional paths are shown in Fig.~\ref{fig:1}C illustrating the two phases of a given rare event path: the exponentially long lag phase (shown in red) and the transition phase (shown in orange).
Note that in general, the lag phase is likely to be much longer than shown.

Recall that we assume the distinguished limit $\epsilon \to 0$, $N\to \infty$ with $\Lambda = \epsilon \log(N)$ constant. 
In regime 1, there is a labor shortage.
If $\Lambda  < V(L)$ then there are not enough walkers to overcome the exponentially long (in $\epsilon^{-1}$) lag phase (see Fig.~\ref{fig:1}C) and $\mathbb{E}[\tau^N_{\rm first}] \to \infty$ as $\epsilon \to 0$.
The MLT is essentially described by the globally optimal $T\to \infty$ MLT \eqref{eq:opt_MLT} (see Fig.~\ref{fig:1}A,B).
Increasing the number of walkers in this regime has the strongest effect on reducing the extreme first hitting time; that is, $\mathbb{E}[\tau^N_{\rm first}] \propto N^{-1}$.
One can define a critical labor supply, 
\begin{equation}
  N_* = \exp(V(L)/\epsilon) = \exp\left(\frac{\mu}{2\epsilon}\Vert L\Vert^2\right),
\end{equation}
which serves as a boundary between the regime 1 and 2.

\begin{figure}[tbp]
    \centering
    \includegraphics[width=6.6cm]{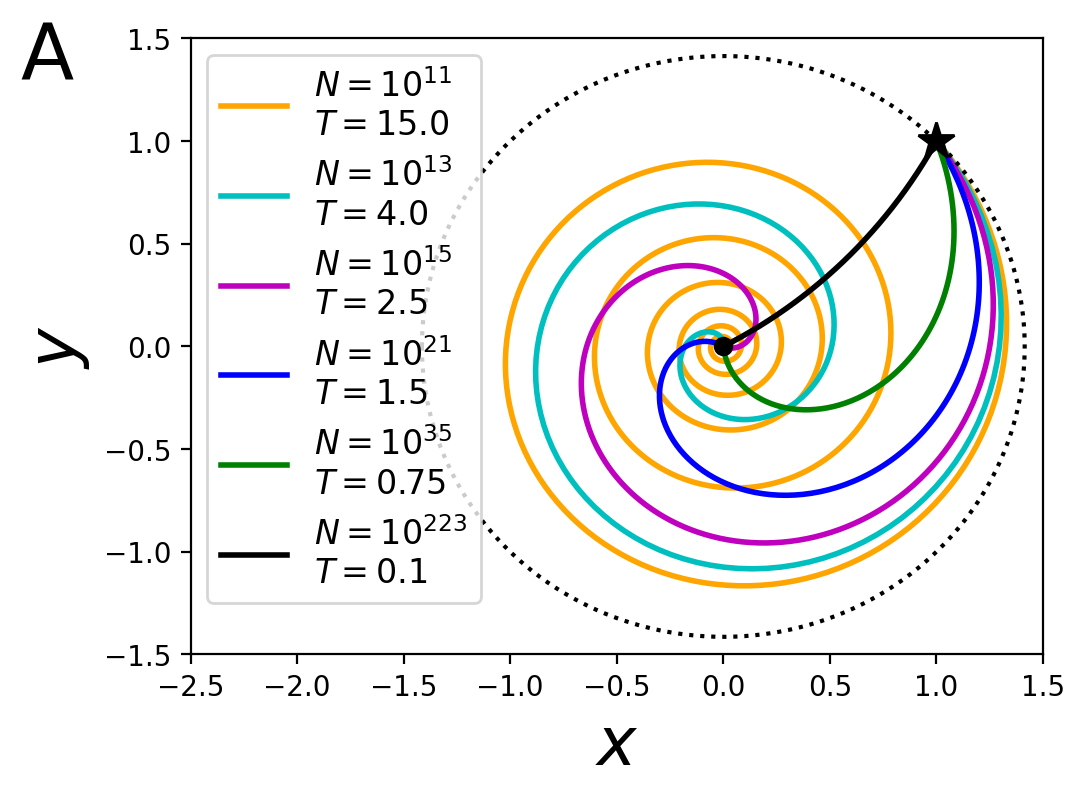}
    \includegraphics[width=8.5cm]{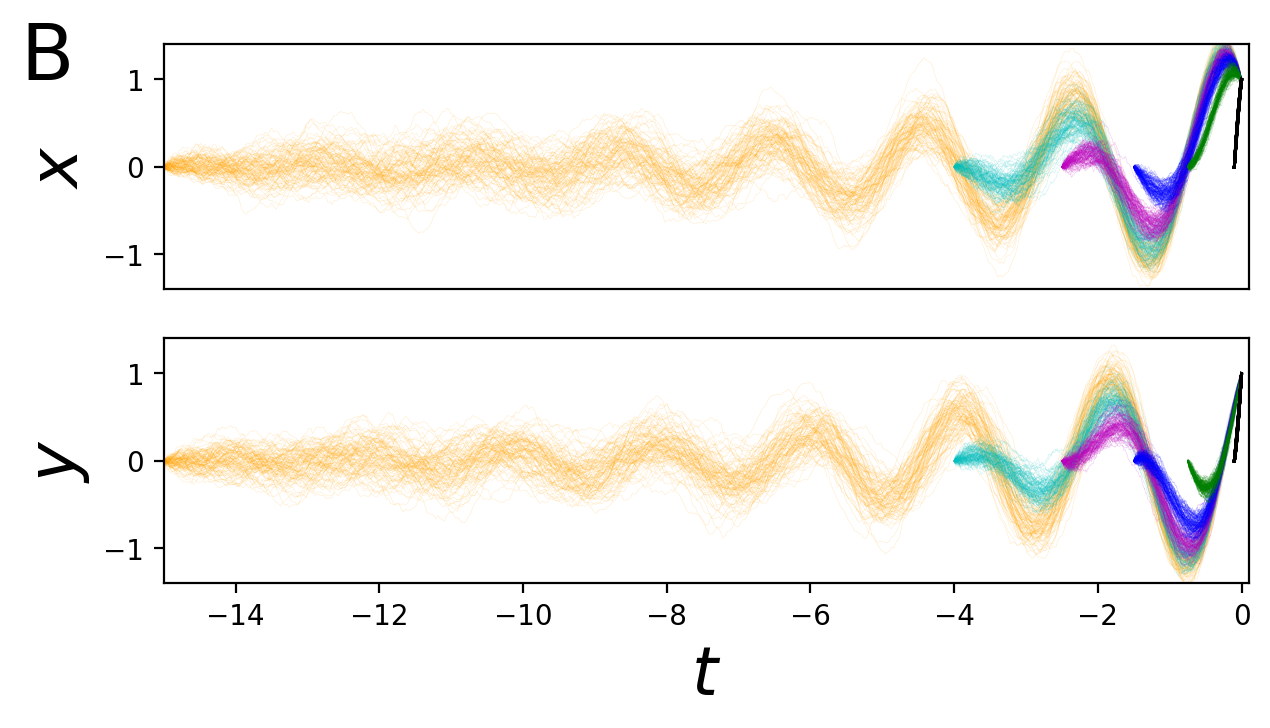}
    \caption{{\bf Trajectories in regime 2: balanced labor.} In Regime 2, the MLTs are deformed as $N$ is increased ($T$ is decreased), with $N = \exp(S_{T|T}/\epsilon)$, $\epsilon=10^{-2}$. (A) Multiple suboptimal MLTs conditioned on the final time $T$. (B) Sampled transition paths for $\epsilon=10^{-2}$. For comparison, the trajectories are aligned at the end and shown backward in time. Other parameters are $\mu=0.25$, $\omega=3$.}
    \label{fig:2}
\end{figure}

In regime 2, the labor supply is balanced against the single-walker rare event time $\tau$.
If $V(L) < \Lambda < S_{T|T} < \infty$ then $N > N_*$ and there is a sufficient number of walkers to accomplish the rare event in finite time as $\epsilon \to 0^+$.
The MLT duration $T$ is finite and does not follow the same path as the globally optimal MLT.~~
We have that $T = T_{\rm ML} = \Psi_L(\Lambda)$, where $\Psi_L(\Lambda)$ is given by \eqref{eq: Psi x c}.
The conditional mean $\varphi_{t|T}$ \eqref{eq:MLT_sol} is plotted in Fig.~\ref{fig:2}A for several values of $N$ (We also note the corresponding value of $T$).
Sample paths are shown in Fig.~\ref{fig:2}B as functions of time.
Increasing the number of walkers is equivalent to reducing $T$, which results in an MLT that moves more directly toward the exit point $L$.
For exceedingly large $N$ (e.g., $N = 10^{223}$ or equivalently $T=0.1$), the path is nearly a straight line (see the black curves in Fig.~\ref{fig:2}).

In regime 3, there is a labor surplus.
As $\Lambda \to \infty$ we reach a regime where adding more walkers has the smallest effect on the extreme first hitting time; namely $\mathbb{E}[\tau^N_{\rm first}] \propto 1/\log(N)$.
One way we could have the limit $\Lambda \to \infty$ is taking $N\to\infty$ with $\epsilon$ fixed.
We know that we have a labor surplus when the cost function begins to scale like $S_{T|T} \simeq M T^{-1}$, as $T\to 0^+$, for some constant $M>0$.
The small $T$ asymptotics of the cost function has been studied for a large class of nonlinear SDEs \cite[Ch.\ 11, p.\ 407]{freidlin2012random}, and it is known that the constant $M$ depends on the Riemannian length between the origin and $L$.
With this scaling, we have that $\mathbb{E}[\tau^N_{\rm first}] = T \simeq \frac{M}{\epsilon \log(N)}$, as $T\to 0^+$.
We therefore conclude that the crossover timescale can be determine by the precise small $T$ asymptotics of $S_{T|T}$.
In our simple example, expanding \eqref{eq:cost_S} around $T=0$ yields
\begin{equation}
  S_{T|T} \simeq \frac{\Vert L\Vert^2}{4}\left( T^{-1} + \mu 
  + O(T)\right), \quad T\to 0^+. 
\end{equation} 
Notice that the leading order term is independent of $\mu$ and $\omega$.
The above expansion suggests that the crossover timescale is $T \ll 1/\mu$
which in turn identifies,
\begin{equation}
 N > N_{**} = \exp\left(\frac{V(L, \mu^{-1})}{\epsilon}\right)
\end{equation}
where the scaling of the extreme hitting time is $1/\log(N)$.
The factor $M/\epsilon = \Vert L\Vert^2/(4\epsilon)$ is the mean first passage time for a single Brownian walker to reach the boundary of the circle of radius $L$, starting from the center.
It is this time that is reduced by a factor of $1/\log(N)$ as $N\to\infty$ so that
\begin{equation}
    \mathbb{E}[\tau^N_{\rm first}] \simeq \frac{\left(\frac{\Vert L\Vert^2}{4\epsilon}\right)}{\log(N)}, \quad N\to\infty.
\end{equation}

\section{Discussion}
It is counter-intuitive that an extreme first passage time is highly likely to be caused by a pathway that is highly unlikely to occur in an individual walker.
In other words, even though individual walkers have far more likely pathways they can take to complete the rare event, a much less likely trajectory causes the extreme first passage time.
This challenges the idea that since the individual walkers initiating the extreme event are identical, one can gain understanding of the extreme event by studying the individuals that initiate it.
For example, to understand the factors that contribute to cancer, we might study the sequence of events (e.g., mutations) that lead an individual cell to a physiological cancer-relevant state.
But, this may not yield an accurate picture of how the first cell initiates cancer in a living organism. 
Likewise, single molecule experiments might observe spontaneous opening events in a calcium channel, but in a cell with thousands of identical channels, the sequence of molecular reconfigurations that ultimately spark a calcium wave could be driven by a very different reaction pathway. 
Whether we seek to understand the behaviors that underlie transmission of a viral contagion during the initiation of an outbreak or how an antibiotic-resistant bacterial population re-establishes itself within a patient that stops taking antibiotics, it is far easier to study the properties of individuals in isolation that it is to study the full population.

The results established in this paper allow us to conclude that we can still employ the strategy of studying the behavior of an individual, but we must consider the effect of conditioning on a maximum event time.
In order to do this, we need to know how the maximum event time depends on the size of the population.
Our results provide this through a simple relationship between the timescale of the extreme first passage time, the strength of the intrinsic noise driving the individual rare events, and the size of the population.

Our results establish a connection between the number of walkers, $N$, and the duration of a MLT, $T$. 
This means that adding more walkers is equivalent to conditioning on a faster {MLT}.
It has long been known in the Large Deviation Theory literature that, given a fixed exit point, the globally optimal MLT is given by $T\to\infty$, and it is these paths that determine the asymptotics of the single walker mean exit time.
Hence, the $T\to \infty$ MLTs have received the most attention in the literature.
By establishing a link between $N$ and $T$, we have identified a new application for the sub-optimal, finite $T$, MLTs.

The example in Section \ref{sec:example} involved a linear SDE. 
For general nonlinear problems, the best course of action is to employ numerical methods \cite{heymann2008geometric,cameron2012finding,newby2014} and dynamical systems techniques \cite{maier1993,maier1996scaling,maier1997,newby2014} to study the problem.
The deterministic ($\epsilon = 0$) dynamics should be studied first.
The stochastic dynamics can then be understood through \eqref{eq:ham_dyn_gen}.
Notable features that can occur in general nonlinear problems include the possibility that optimal ($T\to\infty$) MLTs are not unique \cite{maier1996scaling,newby2014} and the "non classical" behavior of MLTs near unstable critical points of the deterministic dynamics \cite{maier1997,newby2014}--- particularly saddles, which often serve as boundaries between basins of attraction.

Finally, we note that our choice of using a continuous Markov process is purely for the sake of brevity. 
The theory of large deviations, including the characterization of rare events in terms of MLTs, has been extended to many other types of Markov process, including continuous time jump processes \cite{freidlin2012random},  piece-wise deterministic Markov processes \cite{freidlin2012random,newby2012isolating,bressloff2014path}, mixtures \cite{newby2014,newby2015bistable,walker2022numerical}, and spatially extended Markov processes \cite{dembo2009large,maclaurin2023phase}.
We emphasize that the results presented here should be generalized by applying LDT directly to the full process in question and not by employing an intermediate diffusion approximation.
This is particularly true for the case where fast variables are averaged out \cite[ch.\ 7]{freidlin2012random}, since this constrains the resulting MLT and leads to a characterization of the rare event using suboptimal trajectories \cite{newby2011asymptotic,newby2013breakdown,newby2013metastable}. 
A general strategy for applying a large deviation analysis to a given weak-noise Markov process is to start by calculating the Hamiltonian function analogous to Eqn.~\eqref{eq:ham_ldt}. 
The most straightforward way to accomplish this is with a WKB expansion \cite{newby2014}.
Once the Hamiltonian is known, one can utilize numerical methods to solve the Hamilton-Jacobi equation \eqref{eq: Hamiltonian PDE} via \eqref{eq:ham_dyn_gen}.


\appendix
\section{Appendix}
\label{sec:example appendix}
We provide some extra details on the Ornstein-Uhlenbeck process of Section \ref{sec:example}.
It is well-known that $X_t$ is Gaussian \cite[Section 5.6]{karatzas1998brownian}, and hence one can immediately infer that, writing $B_{\delta}(L) \subset \mathbb{R}^2$ to be a closed ball of radius $\delta$ about $L$,
\begin{align}
\epsilon \log \mathbb{P}\big(x_t \in  B_{\delta}(L) \big) = - \frac{1}{4} L^T U_{tt}^{-1} L + o(\epsilon,\delta),
\end{align}
where the corrections terms are negligible in taking $\epsilon \to 0^+$ (first), and then $\delta \to 0^+$ (second). We now claim that the above quadratic form is in fact the cost-function, and that the solution of \eqref{eq:ham_dyn}-\eqref{eq: boundary condition OU example} given by \eqref{eq:MLT_sol}, is the conditional mean (defined just below in \eqref{eq: conditional mean definition}).

We first establish the second statement as follows. 
For any $s\in [0,t]$, define $m_{s|t} \in \mathbb{R}^2$ to be the conditional mean, that is
\begin{align}
m_{s|t} &= \mathbb{E}\left[ x_s \mid x_t = L,\, x_0 = 0 \right] \label{eq: conditional mean definition} = U_{st} U_{tt}^{-1}L  .
\end{align}
One can easily verify that the above expression is consistent with \eqref{eq:MLT_sol}.

Next, we justify the claim that the cost-function assumes the form \eqref{eq: OU cost function}. That is, (for any $L \in \mathbb{R}^2$ and any $t > 0$)
\begin{align} \label{eq: to show cost function}
V(L,t) = \frac{1}{2} L^T  U_{tt}^{-1} L.
\end{align}
To see why this holds, for any finite subset $\pi \subset [0,t]$, define the event
\begin{align}
\mathcal{A}_{\pi, \tilde{\delta}} = \big\lbrace y\in \mathcal{C}([0,t],\mathbb{R}^2) , \;  \sup_{s\in \pi} \| y_s - m_{s|t} \| \geq \tilde{\delta} \big\rbrace .
\end{align}
We can thus conclude that for any $\tilde{\delta} > 0$,
\begin{multline}
\lsup{\delta} \lsup{\epsilon}\epsilon \log \mathbb{P}\big(x_t \in B_{\delta}(L) \big) - \lsup{\delta}\lsup{\epsilon}\epsilon \log \mathbb{P}\big(x_t \in B_{\delta}(L) , x \in \mathcal{A}_{\pi, \tilde{\delta}}  \big) \\
= \frac{1}{2}\sup_{y\in \mathcal{A}_{\pi, \tilde{\delta}} } \sum_{r,s\in \pi} \sum_{a,b \in \lbrace 1,2 \rbrace} y_r(a) y_s(b)  \tilde{V}_{rs}(a,b) > 0.
 \label{eq: conditional mean LDP}
\end{multline}
where $\tilde{V} = \big( \tilde{V}_{rs}(a,b) \big)_{r,s\in \pi, a,b \in \lbrace 1,2 \rbrace} \in \mathbb{R}^{2|\pi| \times 2 | \pi | }$ is the covariance matrix with elements $\big( U_{rs | t}(a,b) \big)_{r,s\in \pi , a,b \in \lbrace 1,2 \rbrace}$. Indeed the RHS of  \eqref{eq: conditional mean LDP} is strictly positive because $\tilde{V}$ must be strictly positive definite (the inverse of a positive definite symmetric matrix is codiagonal with the original matrix, with each eigenvalues the inverse of one of the original eigenvalues).

It follows from \eqref{eq: conditional mean LDP} that 
\begin{align}
\lim_{\delta \to 0^+}\lim_{\epsilon \to 0^+} \epsilon \log \mathbb{P}\big( x_t \in B_{\delta}(L) \big) = \lim_{\delta \to 0^+}\lim_{\epsilon \to 0^+} \epsilon \log \mathbb{P}\big( \sup_{s\in [0,t]}\| x_s - \phi_{s|t} \| \leq \delta \big).
\end{align}
Because the Large Deviations Rate Function $I_T$ is lower-semi-continuous, this implies that
\begin{align}
\inf_{z\in H_1\; : z_t = L}I_T(z) = I_T(\phi).
\end{align}
which immediately implies \eqref{eq: to show cost function}.

It is also known that \cite{karatzas1998brownian} the covariance matrix satisfies the ODE
\begin{align}
\label{eq:cov_tt}
\frac{d}{dt}U_{tt} = A U_{tt} + U_{tt}A^T + 2 \mathbb{I}, \quad U_{00} = 0 \mathbb{I}.
\end{align}
Because the eigenvalues of $A$ have negative real parts, there must exist a unique invariant distribution to the stochastic dynamics \eqref{eq:sde_ou}. This distribution is centered and Gaussian, with covariance matrix $U_* \in \mathbb{R}^{2\times 2}$ satisfying the identity
\begin{align}
\label{eq:cov steady state}
A U_* + U_* A^T + 2\mathbb{I} = 0 .
\end{align}
That is, there is a unique positive-definite symmetric matrix $U_*$ satisfying the above identity, and $U_{tt}$ converges to $U_*$ as $t\to\infty$. We thus find that the quasi-potential is
\begin{align}
V(L) = \frac{1}{2} L^T  U_{*}^{-1}L.
\end{align}

\end{document}